\newif{\ifarxiv}
\arxivtrue%

\documentclass[orivec,envcountsect]{llncs}

\usepackage{silence}
\usepackage{amsmath}
\usepackage{amssymb}

\usepackage{arydshln}
\usepackage{amsthm}
\usepackage{stmaryrd}
\usepackage{mathpartir}
\usepackage{upgreek}
\usepackage{microtype}

\usepackage{tikz}
\usetikzlibrary{positioning}

\usepackage{aliascnt}
\usepackage[colorlinks,urlcolor=black,citecolor=black,linkcolor=black]{hyperref}

\usepackage{paralist}
\usepackage{xspace}

\usepackage{thm-restate}
\usepackage{cleveref}

\SetSymbolFont{stmry}{bold}{U}{stmry}{m}{n}

\newcommand{\alphabet}{\Sigma}
\newcommand{\terms}{\ensuremath{{\mathcal{T}}}}

\newcommand{\KA}{\ensuremath{\mathsf{KA}}\xspace}
\newcommand{\BKA}{\ensuremath{\mathsf{BKA}}\xspace}
\newcommand{\CKA}{\ensuremath{\mathsf{CKA}}\xspace}
\newcommand{\AKA}{\ensuremath{\mathsf{T}}\xspace}

\newcommand{\equivt}[1]{\equiv_{\scriptscriptstyle#1}}
\newcommand{\leqqt}[1]{\leqq_{\scriptscriptstyle#1}}
\newcommand{\geqqt}[1]{\geqq_{\scriptscriptstyle#1}}
\newcommand{\semt}[2]{{\left\llbracket#2\right\rrbracket}_{\scriptscriptstyle#1}}

\newcommand{\equivbka}{\equivt{\BKA}}
\newcommand{\equivcka}{\equivt{\CKA}}
\newcommand{\equivaka}{\equivt{\AKA}}
\newcommand{\leqqka}{\leqqt{\KA}}
\newcommand{\leqqbka}{\leqqt{\BKA}}
\newcommand{\leqqcka}{\leqqt{\CKA}}
\newcommand{\leqqaka}{\leqqt{\AKA}}
\newcommand{\geqqaka}{\geqqt{\AKA}}
\newcommand{\sembka}{\semt{\BKA}}
\newcommand{\semcka}{\semt{\CKA}}
\newcommand{\semaka}{\semt{\AKA}}

\newcommand{\angl}[1]{\left\langle#1\right\rangle}
\newcommand{\pipe}{\;\;|\;\;}
\newcommand{\width}[1]{|#1|}

\newcommand{\lp}[1]{\mathbf{#1}}
\newcommand{\pl}[1]{\mathcal{#1}}
\newcommand{\restr}[1]{\mskip-1mu\upharpoonright_{#1}}

\newcommand{\naturals}{\mathbb{N}}
\newcommand{\sppom}{\mathsf{SP}}
\newcommand{\down}[1]{#1{\downarrow}}

\newcommand{\psplit}[1]{\mathrel{\Updelta_{#1}}}
\newcommand{\ssplit}[1]{\mathrel{\nabla_{#1}}}
\newcommand{\lefts}{\ell}
\newcommand{\rights}{r}

\newcommand{\N}{\ensuremath{\mathsf{N}}\xspace}

\newcommand{\pc}{\odot}
\newcommand{\ic}{\otimes}

\newcommand{\ls}[1]{\mathfrak{#1}}

\title{Concurrent Kleene Algebra: \texorpdfstring{\\}{} Free Model and Completeness}

\author{%
  Tobias Kapp\'{e}\and
  Paul Brunet\and
  Alexandra Silva\and
  Fabio Zanasi%
}
\institute{University College London}

\ifarxiv%
\else%
\fi%

\ifarxiv%
\newcommand{\slpaper}{\cite{laurence-struth-2017-arxiv}}
\else%
\newcommand{\slpaper}{\cite{laurence-struth-2017-arxiv-fixedurl}}
\fi%

\begin{document}

\maketitle

\begin{abstract}
Concurrent Kleene Algebra (CKA) was introduced by Hoare, Moeller, Struth and Wehrman in 2009 as a framework to reason about concurrent programs.
We prove that the axioms for CKA with bounded parallelism are complete for the semantics proposed in the original paper; consequently, these semantics are the free model for this fragment.
This result settles a conjecture of Hoare and collaborators.
Moreover, the technique developed to this end allows us to establish a Kleene Theorem for CKA, extending an earlier Kleene Theorem for a fragment of CKA\@.
 \end{abstract}

\section{Introduction}%
\label{section:introduction}

Concurrent Kleene Algebra (\CKA)~\cite{hoare-moeller-struth-wehrman-2009} is a mathematical formalism which extends Kleene Algebra (\KA) with a parallel composition operator, in order to express concurrent program behaviour.%
\footnote{In its original formulation, \CKA also features an operator (\emph{parallel star}) for unbounded parallelism: in harmony with several recent works~\cite{jipsen-moshier-2016,kappe-brunet-luttik-silva-zanasi-2017}, we study the variant of \CKA without parallel star, sometimes called ``weak'' \CKA.}
In spite of such a seemingly simple addition, extending the existing \KA toolkit (notably, completeness) to the setting of \CKA turned out to be a challenging task.
A lot of research happened since the original paper, both foundational~\cite{laurence-struth-2014,jipsen-moshier-2016} and on how \CKA could be used to reason about important verification tasks in concurrent systems~\cite{horn-kroening-2015,hoare-staden-moeller-struth-zhu-2016}.
However, and despite several conjectures~\cite{hoare-staden-moeller-struth-zhu-2016,jipsen-moshier-2016}, the question of the characterisation of the free \CKA and the completeness of the axioms remained open, making it impractical to use \CKA in verification tasks.
This paper settles these two open questions.
We answer positively the conjecture that the free model of \CKA is formed by  series parallel pomset languages, downward-closed under Gischer's subsumption order~\cite{gischer-1988} --- a generalisation of regular languages to sets of partially ordered words.
To this end, we prove that the original axioms proposed in~\cite{hoare-moeller-struth-wehrman-2009} are indeed complete.

Our proof of completeness is based on extending an existing completeness result that establishes series-parallel rational pomset languages as the free Bi-Kleene Algebra (\BKA)~\cite{laurence-struth-2014}.
The extension to the existing result for \BKA provides a clear understanding of the difficulties introduced by the presence of the exchange axiom and shows how to separate concerns between \CKA and \BKA, a technique also useful elsewhere.
For one, our construction also provides an extension of (half of) Kleene's theorem for \BKA~\cite{kappe-brunet-luttik-silva-zanasi-2017} to \CKA, establishing pomset automata as an operational model for \CKA and opening the door to decidability procedures similar to those previously studied for \KA. 
Furthermore, it reduces deciding the equational theory of \CKA to deciding the equational theory of \BKA. 

\BKA is defined as \CKA with the only (but significant) omission of the \emph{exchange law}, $(e \parallel f) \cdot (g \parallel h) \leqqcka (e \cdot g) \parallel (f \cdot h)$.
The exchange law is the core element of \CKA as it softens true concurrency: it states that when two sequentially composed programs (i.e., $e \cdot g$ and $f \cdot h$) are composed in parallel, they can be implemented by running their heads in parallel, followed by running their tails in parallel (i.e., $e \parallel f$, then $g \parallel h$).
The exchange law allows the implementer of a \CKA expression to interleave threads at will, without violating the specification.

To illustrate the use of the exchange law, consider a protocol with three actions: query a channel $c$, collect an answer from the same channel, and print an unrelated message $m$ on screen.
The specification for this protocol requires the query to happen before reception of the message, but the printing action being independent, it may be executed concurrently.
We will write this specification as~$\left(q(c)\cdot r(c)\right)\parallel p(m)$, with the operator $\cdot$ denoting sequential composition.
However, if one wants to implement this protocol in a sequential programming language, a total ordering of these events has to be introduced.
Suppose we choose to implement this protocol by printing $m$ while we wait to receive an answer.
This implementation can be written $q(c)\cdot p(m)\cdot r(c)$.
Using the laws of \CKA, we can prove that $q(c)\cdot p(m)\cdot r(c)\leqqcka\left(q(c)\cdot r(c)\right)\parallel p(m)$, which we interpret as the fact that this implementation respects the specification.
Intuitively, this means that the specification lists the necessary dependencies, but the implementation can introduce more.

Having a complete axiomatisation of \CKA has two main benefits.
First, it allows one to get certificates of correctness.
Indeed, if one wants to use \CKA for program verification, the decision procedure presented in~\cite{brunet-pous-struth-2017} may be used to test program equivalence.
If the test gives a negative answer, this algorithm provides a counter-example.
However if the answer is positive, no meaningful witness is produced.
With the completeness result presented here, that is constructive in nature, one could generate an axiomatic proof of equivalence in these cases.
Second, it gives one a simple way of checking when the aforementioned procedure applies.
By construction, we know that two terms are semantically equivalent whenever they are equal in every concurrent Kleene algebra, that is any model of the axioms of \CKA. 
This means that if we consider a specific semantic domain, one simply needs to check that the axioms of \CKA hold in there to know that the decision procedure of~\cite{brunet-pous-struth-2017} is sound in this model.

While this paper was in writing, a manuscript with the same result appeared~\slpaper{}.
Among other things, the proof presented here is different in that it explicitly shows how to syntactically construct terms that express certain pomset languages, as opposed to showing that such terms must exist by reasoning on a semantic level.
We refer to~\Cref{section:discussion-further-work} for a more extensive comparison.

The remainder of this paper is organised as follows. In Section~\ref{section:overview}, we give an informal overview of the completeness proof.
In Section~\ref{section:preliminaries}, we introduce the necessary concepts, notation and lemmas.
In Section~\ref{section:axiomatising}, we work out the proof.
We discuss the result in a broader perspective and outline further work in Section~\ref{section:discussion-further-work}.

\section{Overview of the Completeness Proof}%
\label{section:overview}

We start with an overview of the steps necessary to arrive at the main result.
As mentioned, our strategy in tackling \CKA-completeness is to build on the existing \BKA-completeness result.
Following an observation by Laurence and Struth, we identify \emph{downward-closure} (under Gischer's subsumption order~\cite{gischer-1988}) as the feature that distinguishes the pomsets giving semantics to \BKA-expressions from those associated with \CKA-expressions.
In a slogan,
\begin{center}
  \CKA-semantics = \BKA-semantics + downward-closure.
\end{center}
This situation is depicted in the upper part of the commuting diagram in \Cref{figure:closure-step}.
Intuitively, downward-closure can be thought of as the semantic outcome of adding the exchange axiom, which distinguishes \CKA from \BKA. 
Thus, if $a$ and $b$ are events that can happen in parallel according to the \BKA-semantics of a term, then $a$ and $b$ may also be ordered in the \CKA-semantics of that same term.

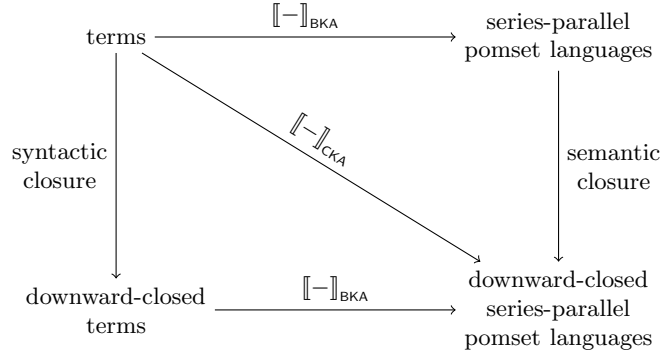
\begin{figure}[h!]
\centering
\begin{tikzpicture}
\node (terms) {terms};
\node[align=center,right=4cm of terms] (sppomlang) {series-parallel \\ pomset languages};
\node[align=center,below=3cm of terms] (dcterms) {downward-closed \\ terms};
\node[align=center] at (dcterms -| sppomlang) (dcsppomlang) {downward-closed \\ series-parallel \\ pomset languages};
\draw[->] (terms) edge node[above] {$\sembka{-}$} (sppomlang);
\draw[->] (sppomlang) edge node[right,align=center] {semantic \\ closure} (dcsppomlang);
\draw[->] (terms) edge node[sloped,above] {$\semcka{-}$} (dcsppomlang);
\draw[->] (terms) edge node[left,align=center] {syntactic \\ closure} (dcterms);
\draw[->] (dcterms) edge node[above] {$\sembka{-}$} (dcsppomlang);
\end{tikzpicture}
\caption{The connection between \BKA and \CKA semantics mediated by closure.}%
\label{figure:closure-step}
\end{figure}

The core of our \CKA-completeness proof will be to construct a syntactic counterpart to the semantic closure.
Concretely, we shall build a function that maps a \CKA term $e$ to an equivalent term $\down{e}$, called the (syntactic) \emph{closure} of $e$.
The lower part of the commuting diagram in \Cref{figure:closure-step} shows the property that $\down{e}$ must satisfy in order to deserve the name of closure: its \BKA semantics has to be the same as the \CKA semantics of $e$.

\begin{example}
Consider $e = a \parallel b$, whose \CKA-semantics prescribe that $a$ and $b$ are events that may happen in parallel.
One closure of this term would be $\down{e} = a \parallel b + a \cdot b + b \cdot a$, whose \BKA-semantics stipulate that either $a$ and $b$ execute purely in parallel, or $a$ precedes $b$, or $b$ precedes $a$ --- thus matching the optional parallelism of $a$ and $b$.
For a more non-trivial example, take $e = a^\star \parallel b^\star$, which represents that finitely many repetitions of $a$ and $b$ occur, possibly in parallel.
A closure of this term would be $\down{e} = {(a^\star \parallel b^\star)}^\star$: finitely many repetitions of $a$ and $b$ occur truly in parallel, which is repeated indefinitely.
\end{example}

In order to find $\down{e}$ systematically, we are going to construct it in stages, through a completely syntactic procedure where each transformation has to be valid according to the axioms. There are three main stages.
\begin{enumerate}[(i)]
    \item
    We note that, not unexpectedly, the hardest case for computing the closure of a term is when $e$ is a parallel composition, i.e., when $e = e_0 \parallel e_1$ for some \CKA terms $e_0$ and $e_1$.
    For the other operators, the closure of the result can be obtained by applying the same operator to the closures of its arguments.
    For instance, $\down{\left(e+f\right)}= \down{e} + \down{f}$.
    This means that we can focus on calculating the closure for the particular case of parallel composition.

    \item
    We construct a \emph{preclosure} of such terms $e$, whose \BKA semantics contains all but possibly the sequentially composed pomsets of the \CKA semantics of $e$.
    Since every sequentially composed pomset decomposes (uniquely) into non-sequential pomsets, we can use the preclosure as a basis for induction.

    \item
    We extend this preclosure of $e$ to a proper closure, by leveraging the fixpoint axioms of \KA to solve a system of linear inequations.
    This system encodes ``stringing together'' non-sequential pomsets to build all pomsets in $e$.
\end{enumerate}

As a straightforward consequence of the closure construction, we obtain a completeness theorem for \CKA, which establishes the set of closed series-rational pomset languages as the free \CKA. 

\section{Preliminaries}%
\label{section:preliminaries}

We fix a finite set of symbols $\alphabet$, the \emph{alphabet}.
We use the symbols $a$, $b$ and $c$ to denote elements of $\alphabet$.
The two-element set $\{0,  1\}$ is denoted by $2$.
Given a set $S$, the set of subsets (\emph{powerset}) of $S$ is denoted by $2^S$.

In the interest of readability, the proofs for technical lemmas in this section
\ifarxiv%
are deferred to \Cref{appendix:proofs-preliminaries}.
\else%
can be found in the full version~\cite{brunet-kappe-silva-zanasi-2017-arxiv}.
\fi%

\subsection{Pomsets}

A trace of a sequential program can be modelled as a word, where each letter represents an atomic event, and the order of the letters in the word represents the order in which the events took place.
Analogously, a trace of a concurrent program can be thought of as word where letters are partially ordered, i.e., there need not be a causal link between events.
In literature, such a partially ordered word is commonly called a \emph{partial word}~\cite{grabowski-1981}, or \emph{partially ordered multiset} (\emph{pomset}, for short)~\cite{gischer-1988}; we use the latter term.

A formal definition of pomsets requires some work, because the partial order should order \emph{occurrences} of events rather than the events themselves.
For this reason, we first define a labelled poset.

\begin{definition}%
\label{definition:lp}
A \emph{labelled poset} is a tuple $\angl{S, \leq, \lambda}$, where $\angl{S, \leq}$ is a partially ordered set (i.e., $S$ is a set and $\leq$ is a partial order on $S$), in which $S$ is called the \emph{carrier} and $\leq$ is the \emph{order}; $\lambda: S \to \alphabet$ is a function called the \emph{labelling}.
\end{definition}
We denote labelled posets with lower-case bold symbols $\lp{u}$, $\lp{v}$, et cetera.
Given a labelled poset $\lp{u}$, we write $S_\lp{u}$ for its carrier, $\leq_\lp{u}$ for its order and $\lambda_\lp{u}$ for its labelling.
We write $\lp{1}$ for the empty labelled poset.
We say that two labelled posets are \emph{disjoint} if their carriers are disjoint.

Disjoint labelled posets can be composed parallelly and sequentially; parallel composition simply juxtaposes the events, while sequential composition imposes an ordering between occurrences of events originating from the left operand and those originating from the right operand.

\begin{definition}%
\label{definition:lp-composition}
Let $\lp{u}$ and $\lp{v}$ be disjoint.
We write $\lp{u} \parallel \lp{v}$ for the \emph{parallel composition} of $\lp{u}$ and $\lp{v}$, which is the labelled poset with the carrier $S_{\lp{u} \cup \lp{v}} = S_\lp{u} \cup S_\lp{v}$, the order $\leq_{\lp{u} \parallel \lp{v}}\ =\ \leq_\lp{u} \cup\ \leq_\lp{v}$ and the labeling $\lambda_{\lp{u} \parallel \lp{v}}$ defined by
\[
    \lambda_{\lp{u} \parallel \lp{v}}(x) =
    \begin{cases}
      \lambda_\lp{u}(x)&x \in S_\lp{u};\\
      \lambda_\lp{v}(x)&x \in S_\lp{v}.
    \end{cases}
\]

Similarly, we write $\lp{u} \cdot \lp{v}$ for the \emph{sequential composition} of $\lp{u}$ and $\lp{v}$, that is, labelled poset with the carrier $S_{\lp{u} \cup \lp{v}}$ and the partial order
\[\leq_{\lp{u} \cdot \lp{v}}\ =\ \leq_\lp{u} \cup \leq_\lp{v} \cup\ (S_\lp{u} \times S_\lp{v}),\]
as well as the labelling $\lambda_{\lp{u} \cdot \lp{v}} = \lambda_{\lp{u} \parallel \lp{v}}$.
\end{definition}
\noindent
Note that $\lp{1}$ is neutral for sequential and parallel composition, in the sense that we have $\lp{1} \parallel \lp{u} = \lp{1} \cdot \lp{u} = \lp{u} = \lp{u} \cdot \lp{1} = \lp{u} \parallel \lp{1}$.

There is a natural ordering between labelled posets with regard to concurrency.

\begin{definition}%
\label{definition:lp-relations}
Let $\lp{u}, \lp{v}$ be labelled posets.
A \emph{subsumption} from $\lp{u}$ to $\lp{v}$ is a bijection $h: S_\lp{u} \to S_\lp{v}$ that preserves order and labels, i.e., $u \leq_\lp{u} u'$ implies that $h(u) \leq_\lp{v} h(u')$, and $\lambda_\lp{v} \circ h = \lambda_\lp{u}$.
We simplify and write $h: \lp{u} \to \lp{v}$ for a subsumption from $\lp{u}$ to $\lp{v}$.
If such a subsumption exists, we write $\lp{v} \sqsubseteq \lp{u}$.
Furthermore, $h$ is an \emph{isomorphism} if both $h$ and its inverse $h^{-1}$ are subsumptions.
If there exists an isomorphism from $\lp{u}$ to $\lp{v}$ we write $\lp{u} \cong \lp{v}$.
\end{definition}

Intuitively, if $\lp{u} \sqsubseteq \lp{v}$, then $\lp{u}$ and $\lp{v}$ both order the same set of (occurrences of) events, but $\lp{u}$ has more causal links, or ``is more sequential'' than $\lp{v}$.
One easily sees that $\sqsubseteq$ is a preorder on labelled posets of finite carrier.

Since the actual contents of the carrier of a labelled poset do not matter, we can abstract from them using isomorphism.
This gives rise to pomsets.

\begin{definition}%
\label{definition:pomset}
A \emph{pomset} is an isomorphism class of labelled posets, i.e., the class $[\lp{v}] \triangleq \{ \lp{u} : \lp{u} \cong \lp{v} \}$ for some labelled poset $\lp{v}$.
Composition lifts to pomsets: we write $[\lp{u}] \parallel [\lp{v}]$ for $[\lp{u} \parallel \lp{v}]$ and $[\lp{u}] \cdot [\lp{v}]$ for $[\lp{u} \cdot \lp{v}]$.
Similarly, subsumption also lifts to pomsets: we write $[\lp{u}] \sqsubseteq [\lp{v}]$, precisely when $\lp{u} \sqsubseteq \lp{v}$.
\end{definition}
We denote pomsets with upper-case symbols $U$, $V$, et cetera.
The \emph{empty pomset}, i.e., $[\lp{1}] = \{ \lp{1} \}$, is denoted by $1$; this pomset is neutral for sequential and parallel composition.
To ensure that $[\lp{v}]$ is a set, we limit the discussion to labelled posets whose carrier is a subset of some set $\mathbb{S}$.
The labelled posets in this paper have finite carrier; it thus suffices to choose $\mathbb{S} = \naturals$ to represent all pomsets with finite (or even countably infinite) carrier.

Composition of pomsets is well-defined: if $\lp{u}$ and $\lp{v}$ are not disjoint, we can find $\lp{u}', \lp{v}'$ disjoint from $\lp{u}, \lp{v}$ respectively such that $\lp{u} \cong \lp{u}'$ and $\lp{v} \cong \lp{v}'$.
The choice of representative does not matter, for if $\lp{u} \cong \lp{u}'$ and $\lp{v} \cong \lp{v'}$, then $\lp{u} \cdot \lp{v} \cong \lp{u}' \cdot \lp{v}'$.
Subsumption of pomsets is also well-defined: if $\lp{u}' \cong \lp{u} \sqsubseteq \lp{v} \cong \lp{v}'$, then $\lp{u}' \sqsubseteq \lp{v}'$.
One easily sees that $\sqsubseteq$ is a partial order on finite pomsets, and that sequential and parallel composition are monotone with respect to $\sqsubseteq$, i.e., if $U \sqsubseteq W$ and $V \sqsubseteq X$, then $U \cdot V \sqsubseteq W \cdot X$ and $U \parallel V \sqsubseteq W \parallel X$.
Lastly, we note that both types of composition are associative, both on the level of pomsets and labelled posets; we therefore omit parentheses when no ambiguity is likely.

\subsubsection{Series-parallel pomsets}

If $a \in \alphabet$, we can construct a labelled poset with a single element labelled by $a$; indeed, since any labelled poset thus constructed is isomorphic, we also use $a$ to denote this isomorphism class; such a pomset is called a \emph{primitive pomset}.
A pomset built from primitive pomsets and sequential and parallel composition is called \emph{series-parallel}; more formally:
\begin{definition}%
\label{definition:pomset-sp}
The set of \emph{series-parallel} pomsets, denoted $\sppom(\alphabet)$, is the smallest set such that $1 \in \sppom(\alphabet)$ as well as $a \in \sppom(\alphabet)$ for every $a \in \alphabet$, and is closed under parallel and sequential composition.
\end{definition}

We elide the sequential composition operator when we explicitly construct a pomset from primitive pomsets, i.e., we write $ab$ instead of $a \cdot b$ for the pomset obtained by sequentially composing the (primitive) pomsets $a$ and $b$.
In this notation, sequential composition takes precedence over parallel composition.

All pomsets encountered in this paper are series-parallel.
A useful feature of series-parallel pomsets is that we can deconstruct them in a standard fashion~\cite{gischer-1988}.

\begin{lemma}%
\label{lemma:pomset-unique-decomposition}
Let $U \in \sppom(\alphabet)$.
Then \emph{exactly one} of the following is true: either
\begin{inparaenum}[(i)]
    \item $U = 1$, or
    \item $U = a$ for some $a \in \alphabet$, or
    \item $U = U_0 \cdot U_1$ for $U_0, U_1 \in \sppom(\alphabet) \setminus \{ 1 \}$, or
    \item $U = U_0 \parallel U_1$ for $U_0, U_1 \in \sppom(\alphabet) \setminus \{ 1 \}$.
\end{inparaenum}
\end{lemma}

In the sequel, it will be useful to refer to pomsets that are \emph{not} of the third kind above, i.e., cannot be written as $U_0 \cdot U_1$ for $U_0, U_1 \in \sppom(\alphabet) \setminus \{ 1 \}$, as \emph{non-sequential} pomsets.
\Cref{lemma:pomset-unique-decomposition} gives a normal form for series-parallel pomsets, as follows.

\begin{corollary}%
\label{lemma:pomset-normal-form}
A pomset $U \in \sppom(\alphabet)$ can be uniquely decomposed as $U = U_0 \cdot U_1 \cdots U_{n-1}$, where for all $0 \leq i < n$, $U_i$ is series parallel and non-sequential.
\end{corollary}

\subsubsection{Factorisation}

We now go over some lemmas on pomsets that will allow us to factorise pomsets later on.
First of all, one easily shows that subsumption is irrelevant on empty and primitive pomsets, as witnessed by the following lemma.
\begin{restatable}{lemma}{pomsetsubsumptionbase}%
\label{lemma:pomset-subsumption-base}
Let $U$ and $V$ be pomsets such that $U \sqsubseteq V$ or $V \sqsubseteq U$.
If $U$ is empty or primitive, then $U = V$.
\end{restatable}

We can also consider how pomset composition and subsumption relate.
It is not hard to see that if a pomset is subsumed by a sequentially composed pomset, then this sequential composition also appears in the subsumed pomset.
A similar statement holds for pomsets that subsume a parallel composition.
\begin{restatable}[Factorisation]{lemma}{pomsetfactorisesubsumption}%
\label{lemma:pomset-factorise-subsumption}
Let $U$, $V_0$, and $V_1$ be pomsets such that $U$ is subsumed by $V_0 \cdot V_1$.
Then there exist pomsets $U_0$ and $U_1$ such that:
\[U = U_0 \cdot U_1,\, U_0 \sqsubseteq V_0,\ and\ U_1 \sqsubseteq V_1.\]

Also, if $U_0$, $U_1$ and $V$ are pomsets such that $U_0 \parallel U_1 \sqsubseteq V$, then there exist pomsets $V_0$ and $V_1$ such that:
\[V = V_0 \parallel V_1,\,U_0 \sqsubseteq V_0,\,and~U_1 \sqsubseteq V_1.\]
\end{restatable}

The next lemma can be thought of as a generalisation of Levi's lemma~\cite{levi-1944}, a well-known statement about words, to pomsets.
It says that if a sequential composition is subsumed by another (possibly longer) sequential composition, then there must be a pomset ``in the middle'', describing the overlap between the two; this pomset gives rise to a factorisation.

\begin{restatable}{lemma}{pomsetlevigeneralised}%
\label{lemma:pomset-levi-generalised}
Let $U$ and $V$ be pomsets, and let $W_0, W_1, \dots, W_{n-1}$ with $n > 0$ be non-empty pomsets such that $U \cdot V \sqsubseteq W_0 \cdot W_1 \cdots W_{n-1}$.
There exists an $m < n$ and pomsets $Y, Z$ such that:
\[Y \cdot Z \sqsubseteq W_m,\,U \sqsubseteq W_0 \cdot W_1 \cdots W_{m-1} \cdot Y,\ and\ V \sqsubseteq Z \cdot W_{m+1} \cdot W_{m+2} \cdots W_n.\]
Moreover, if $U$ and $V$ are series-parallel, then so are $Y$ and $Z$.
\end{restatable}

Levi's lemma also has an analogue for parallel composition.

\begin{restatable}{lemma}{pomsetleviparallel}%
\label{lemma:pomset-levi-parallel}
Let $U, V, W, X$ be pomsets such that $U \parallel V = W \parallel X$.
There exist pomsets $Y_0, Y_1, Z_0, Z_1$ such that
\[U = Y_0 \parallel Y_1,\, V = Z_0 \parallel Z_1,\, W = Y_0 \parallel Z_0,\ and\ X = Y_1 \parallel Z_1.\]
\end{restatable}

The final lemma is useful when we have a sequentially composed pomset subsumed by a parallelly composed pomset.
It tells us that we can factor the involved pomsets to find subsumptions between smaller pomsets.
This lemma first appeared in~\cite{gischer-1988}, where it is called the interpolation lemma.

\begin{restatable}[Interpolation]{lemma}{pomsetinterpolation}%
\label{lemma:pomset-interpolation}
Let $U, V, W, X$ be pomsets such that $U \cdot V$ is subsumed by $W \parallel X$.
Then there exist pomsets $W_0, W_1, X_0, X_1$ such that
\[W_0 \cdot W_1 \sqsubseteq W,\, X_0 \cdot X_1 \sqsubseteq X,\, U \sqsubseteq W_0 \parallel X_0,\ and\ V \sqsubseteq W_1 \parallel X_1.\]
Moreover, if $W$ and $X$ are series-parallel, then so are $W_0$, $W_1$, $X_0$ and $X_1$.
\end{restatable}

On a semi-formal level, the interpolation lemma can be understood as follows.
If $U \cdot V \sqsubseteq W \parallel X$, then the events in $W$ are partitioned between those that end up in $U$, and those that end up in $V$; these give rise to the ``sub-pomsets'' $W_0$ and $W_1$ of $W$, respectively.
Similarly, $X$ partitions into ``sub-pomsets'' $X_0$ and $X_1$.
We refer to \Cref{figure:interpolation} for a graphical depiction of this situation.

Now, if $y$ precedes $z$ in $W_0 \parallel X_0$, then $y$ must precede $z$ in $W \parallel X$, and therefore also in $U \cdot V$.
Since $y$ and $z$ are both events in $U$, it then follows that $y$ precedes $z$ in $U$, establishing that $U \sqsubseteq W_0 \parallel X_0$.
Furthermore, if $y$ precedes $z$ in $W$, then we can exclude the case where $y$ is in $W_1$ and $z$ in $W_0$, for then $z$ precedes $y$ in $U \cdot V$, contradicting that $y$ precedes $z$ in $U \cdot V$.
Accordingly, either $y$ and $z$ both belong to $W_0$ or $W_1$, or $y$ is in $W_0$ while $z$ is in $W_1$; in all of these cases, $y$ must precede $z$ in $W_0 \cdot W_1$.
The other subsumptions hold analogously.
\begin{figure}
    \centering
    \begin{tikzpicture}
    \draw (0,0) rectangle (1,2.2);
    \draw (1.2,0) rectangle (2.2,2.2);
    \draw (3.2,0) rectangle (5.4,1);
    \draw (3.2,1.2) rectangle (5.4,2.2);
    \draw[dotted] (0.1, 0.1) rectangle (0.9,1.05);
    \draw[dotted] (0.1, 1.15) rectangle (0.9,2.1);
    \draw[dotted] (1.3, 0.1) rectangle (2.1,1.05);
    \draw[dotted] (1.3, 1.15) rectangle (2.1,2.1);
    \draw[dotted] (3.3, 0.1) rectangle (4.25,0.9);
    \draw[dotted] (4.35, 0.1) rectangle (5.3,0.9);
    \draw[dotted] (3.3, 1.3) rectangle (4.25,2.1);
    \draw[dotted] (4.35, 1.3) rectangle (5.3,2.1);
    \node[anchor=center] at (0.50,0.55) {$W_0$};
    \node[anchor=center] at (0.50,1.6) {$X_0$};
    \node[anchor=center] at (1.7,0.55) {$W_1$};
    \node[anchor=center] at (1.7,1.6) {$X_1$};
    \node[anchor=center] at (3.75,0.50) {$W_0$};
    \node[anchor=center] at (4.85,0.50) {$W_1$};
    \node[anchor=center] at (3.75,1.7) {$X_0$};
    \node[anchor=center] at (4.85,1.7) {$X_1$};
    \node[anchor=center] (subsume) at (2.7, 1.1) {$\sqsubseteq$};
    \node[anchor=center] (U) at (0.5, 2.5) {$U$};
    \node[anchor=center] (V) at (1.7, 2.5) {$V$};
    \node[anchor=center] (W) at (5.7, 0.5) {$W$};
    \node[anchor=center] (X) at (5.7, 1.7) {$X$};
    \end{tikzpicture}
    \caption{Splitting pomsets in the interpolation lemma}\label{figure:interpolation}
\end{figure}
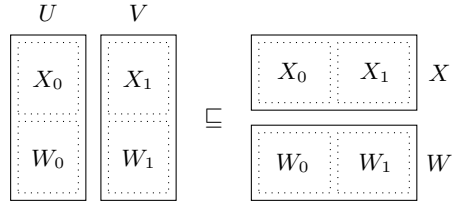

\subsubsection{Pomset languages}

The semantics of \BKA and \CKA are given in terms of sets of series-parallel pomsets.
\begin{definition}%
\label{definition:pl}
A subset of $\sppom(\alphabet)$ is referred to as a \emph{pomset language}.
\end{definition}

As a convention, we denote pomset languages by the symbols $\pl{U}$, $\pl{V}$, et cetera.
Sequential and parallel composition of pomsets extends to pomset languages in a pointwise manner, i.e.,
\[\pl{U} \cdot \pl{V} \triangleq \{ U \cdot V : U \in \pl{U}, V \in \pl{V} \} \]
and similarly for parallel composition.
Like languages of words, pomset languages have a Kleene star operator, which is similarly defined, i.e., $\pl{U}^\star \triangleq \bigcup_{n \in \naturals} \pl{U}^n$, where the $n^{th}$ power of $\pl{U}$ is inductively defined as $\pl{U}^0 \triangleq \{ 1 \}$ and $\pl{U}^{n+1} \triangleq \pl{U}^n \cdot \pl{U}$.

A pomset language $\pl{U}$ is \emph{closed under subsumption} (or simply \emph{closed}) if whenever $U \in \pl{U}$ with $U' \sqsubseteq U$ and $U' \in \sppom(\alphabet)$, it holds that $U' \in \pl{U}$.
The \emph{closure under subsumption} (or simply \emph{closure}) of a pomset language $\pl{U}$, denoted $\down{\pl{U}}$, is defined as the smallest pomset language that contains $\pl{U}$ and is closed, i.e.,
\[\down{\pl{U}} \triangleq \{ U' \in \sppom(\alphabet) : \exists U \in \pl{U}.\ U' \sqsubseteq U \}\]
Closure relates to union, sequential composition and iteration as follows.

\begin{lemma}%
\label{lemma:pl-composition-vs-closure}
Let $\pl{U}, \pl{V}$ be pomset languages; then:
\[\down{(\pl{U} \cup \pl{V})} = \down{\pl{U}} \cup \down{\pl{V}},\, \down{(\pl{U} \cdot \pl{V})} = \down{\pl{U}} \cdot \down{\pl{V}},\ and\ \down{\pl{U}^\star} = \down{\pl{U}}^\star.\]
\end{lemma}
\begin{proof}
The first claim holds for infinite unions, too, and follows immediately from the definition of closure.

For the second claim, suppose that $U \in \pl{U}$ and $V \in \pl{V}$, and that $W \sqsubseteq U \cdot V$.
By \Cref{lemma:pomset-factorise-subsumption}, we find pomsets $W_0$ and $W_1$ such that $W = W_0 \cdot W_1$, with $W_0 \sqsubseteq U$ and $W_1 \sqsubseteq V$.
It then holds that $W_0 \in \down{\pl{U}}$ and $W_1 \in \down{\pl{V}}$, meaning that $W = W_0 \cdot W_1 \in \down{\pl{U}} \cdot \down{\pl{V}}$.
This shows that $\down{(\pl{U} \cdot \pl{V})} \sqsubseteq \down{\pl{U}} \cdot \down{\pl{V}}$.
Proving the reverse inclusion is a simple matter of unfolding the definitions.

For the third claim, we can calculate directly using the first and second parts of this lemma:
\[
\down{\pl{U}^\star}
    = \down{\Bigl( \bigcup_{n \in \naturals} \underbrace{\pl{U} \cdot \pl{U} \cdots \pl{U}}_{n\,\mathrm{\scriptstyle{}times}}\Bigr)}
    = \bigcup_{n \in \naturals} \down{\Bigl(\underbrace{\pl{U} \cdot \pl{U} \cdots \pl{U}}_{n\,\mathrm{\scriptstyle{}times}}\Bigr)}
    = \bigcup_{n \in \naturals} \underbrace{\down{\pl{U}} \cdot \down{\pl{U}} \cdots \down{\pl{U}}}_{n\,\mathrm{\scriptstyle{}times}}
    = \down{\pl{U}}^\star
    \qedhere
\]
\end{proof}

\subsection{Concurrent Kleene Algebra}

We now consider two extensions of Kleene Algebra (\KA), known as \emph{Bi-Kleene Algebra} (\BKA) and \emph{Concurrent Kleene Algebra} (\CKA).
Both extend \KA with an operator for parallel composition and thus share a common syntax.

\begin{definition}%
\label{definition:syntax}
The set $\terms$ is the smallest set generated by the grammar
\[e, f ::= 0 \pipe 1 \pipe a \in \alphabet \pipe e + f \pipe e \cdot f \pipe e \parallel f \pipe e^\star\]
\end{definition}

The \BKA-semantics of a term is a straightforward inductive application of the operators on the level of pomset languages.
The \CKA-semantics of a term is the \BKA-semantics, downward-closed under the subsumption order; the \CKA-semantics thus includes all possible sequentialisations.

\begin{definition}%
\label{definition:semantics}
The function $\sembka{-}: \terms \to 2^{\sppom(\alphabet)}$ is defined as follows:
\begin{align*}
\sembka{0} &\triangleq \emptyset & \sembka{e + f}         &\triangleq \sembka{e} \cup \sembka{f}      & \sembka{e^\star} &\triangleq \sembka{e}^\star  \\
\sembka{1} &\triangleq \{ 1 \}   & \sembka{e \cdot f}     &\triangleq \sembka{e} \cdot \sembka{f}                                                \\
\sembka{a} &\triangleq \{ a \}   & \sembka{e \parallel f} &\triangleq \sembka{e} \parallel \sembka{f}
\end{align*}
Finally, $\semcka{-}: \terms \to 2^{\sppom(\alphabet)}$ is defined as $\semcka{e} \triangleq \down{\sembka{e}}$.
\end{definition}
Following Lodaya and Weil~\cite{lodaya-weil-2000}, if $\pl{U}$ is a pomset language such that $\pl{U} = \sembka{e}$ for some $e \in \terms$, we say that the language $\pl{U}$ is \emph{series-rational}.
Note that if $\pl{U}$ is such that $\pl{U} = \semcka{e}$ for some term $e \in \terms$, then $\pl{U}$ is closed by definition.

To axiomatise semantic equivalence between terms, we build the following relations, which match the axioms proposed in~\cite{laurence-struth-2014}.
The axioms of \CKA as defined in~\cite{hoare-moeller-struth-wehrman-2009} come from a double quantale structure mediated by the exchange law; these imply the ones given here.
The converse implication does not hold; in particular, our syntax does not include an infinitary greatest lower bound operator.
However, \BKA (as defined in this paper) does have a \emph{finitary} greatest lower bound~\cite{laurence-struth-2014}, and by the existence of closure, so does \CKA. 
\begin{definition}%
\label{definition:equivalence}
The relation $\equivbka$ is the smallest congruence on $\terms$ (with respect to all operators) such that for all $e, f, g \in \terms$:
\begin{mathpar}
e + 0 \equivbka e \and
e + e \equivbka e \and
e + f \equivbka f + e \and
e + (f + g) \equivbka (f + g) + h \\
e \cdot 1 \equivbka e\and
1 \cdot e\equivbka e \and
e \cdot (f \cdot g) \equivbka (e \cdot f) \cdot g \\
e \cdot 0 \equivbka 0 \equivbka 0 \cdot e \and
e \cdot (f + g) \equivbka e \cdot f + e \cdot h \and
(e + f) \cdot g \equivbka e \cdot g + f \cdot g \\
e \parallel f \equivbka f \parallel e \and
e \parallel 1 \equivbka e \and
e \parallel (f \parallel g) \equivbka (e \parallel f) \parallel g \\
e \parallel 0 \equivbka 0 \and
e \parallel (f + g) \equivbka e \parallel f + e \parallel g \and
1 + e \cdot e^\star \equivbka e^\star \\
e + f \cdot g \leqqbka g \implies f^\star \cdot e \leqqbka g
\end{mathpar}
in which we use $e \leqqbka f$ as a shorthand for $e + f \equivbka f$.
The final (conditional) axiom is referred to as the \emph{least fixpoint axiom}.

The relation $\equivcka$ is the smallest congruence on $\terms$ that satisfies the rules of $\equivbka$, and furthermore satisfies the \emph{exchange law} for all $e, f, g, h \in \terms$:
\[ (e \parallel f) \cdot (g \parallel h) \leqqcka (e \cdot g) \parallel (f \cdot h) \]
where we similarly use $e \leqqcka f$ as a shorthand for $e + f \equivcka f$.
\end{definition}

We can see that $\equivbka$ includes the familiar axioms of \KA, and stipulates that $\parallel$ is commutative and associative with unit $1$ and annihilator $0$, as well as distributive over $+$.
When using \CKA to model concurrent program flow, the exchange law models sequentialisation: if we have two programs, the first of which executes $e$ followed by $g$, and the second of which executes $f$ followed by $h$, then we can sequentialise this by executing $e$ and $f$ in parallel, followed by executing $g$ and $h$ in parallel.

We use the symbol \AKA in statements that are true for $\AKA \in \{ \BKA, \CKA \}$.
The relation $\equivaka$ is sound for equivalence of terms under \AKA~\cite{jipsen-moshier-2016}. 

\begin{lemma}%
\label{lemma:soundness}
Let $e, f \in \terms$. If $e \equivaka f$, then $\semaka{e} = \semaka{f}$.
\end{lemma}

Since all binary operators are associative (up to $\equivaka$), we drop parentheses when writing terms like $e + f + g$ --- this does not incur ambiguity with regard to $\semaka{-}$.
We furthermore consider $\cdot$ to have precedence over $\parallel$, which has precedence over $+$; as usual, the Kleene star has the highest precedence of all operators.
For instance, when we write $e + f \cdot g^\star \parallel h$, this should be read as $e + ((f \cdot \left(g^\star\right)) \parallel h)$.

In case of \BKA, the implication in \Cref{lemma:soundness} is an equivalence~\cite{laurence-struth-2014}, and thus gives a complete axiomatisation of semantic \BKA-equivalence of terms.%
\footnote{%
    Strictly speaking, the proof in~\cite{laurence-struth-2014} includes the parallel star operator in \BKA. 
    Since this is a conservative extension of \BKA, this proof applies to \BKA as well. 
}

\begin{theorem}%
\label{theorem:wbka-completeness}
Let $e, f \in \terms$. Then $e \equivbka f$ if and only if $\sembka{e} = \sembka{f}$.
\end{theorem}

Given a term $e \in \terms$, we can determine syntactically whether its (\BKA or \CKA) semantics contains the empty pomset, using the function defined below.
\begin{definition}%
\label{definition:nullable}
The \emph{nullability function} $\epsilon: \terms \to 2$ is defined as follows:
\begin{align*}
\epsilon(0) &\triangleq 0 & \epsilon(e + f)         &\triangleq \epsilon(e) \vee \epsilon(f)   & \epsilon(e^\star) \triangleq 1 \\
\epsilon(1) &\triangleq 1 & \epsilon(e \cdot f)     &\triangleq \epsilon(e) \wedge \epsilon(f)                              \\
\epsilon(a) &\triangleq 0 & \epsilon(e \parallel f) &\triangleq \epsilon(e) \wedge \epsilon(f)
\end{align*}
in which $\vee$ and $\wedge$ are understood as the usual lattice operations on $2$.
\end{definition}

That $\epsilon$ encodes the presence of $1$ in the semantics is witnessed by the following.

\begin{restatable}{lemma}{nullable}%
\label{lemma:nullable}
Let $e \in \terms$. Then $\epsilon(e) \leqqaka e$ and $1 \in \semaka{e}$ if and only if $\epsilon(e) = 1$.
\end{restatable}

In the sequel, we need the \emph{(parallel) width} of a term.
This is defined as follows.

\begin{definition}%
\label{definition:width}
Let $e \in \terms$.
The \emph{(parallel) width} of $e$, denoted by $\width{e}$, is defined as $0$ when $e \equivbka 0$; for all other cases, it is defined inductively, as follows:
\begin{align*}
\width{1} &\triangleq 0 & \width{e + f}     &\triangleq \max(\width{e}, \width{f}) & \width{e \parallel f} &\triangleq \width{e} + \width{f} \\
\width{a} &\triangleq 1 & \width{e \cdot f} &\triangleq \max(\width{e}, \width{f}) & \width{e^\star}           &\triangleq \width{e}
\end{align*}
\end{definition}

The width of a term is invariant with respect to equivalence of terms.
\begin{restatable}{lemma}{widthvsequivalence}%
\label{lemma:width-vs-equivalence}
Let $e, f \in \terms$.
If $e \equivbka f$, then $\width{e} = \width{f}$.
\end{restatable}

The width of a term is related to its semantics as demonstrated below.
\begin{restatable}{lemma}{widthvsnonempty}%
\label{lemma:width-vs-nonempty}
Let $e \in \terms$, and let $U \in \sembka{e}$ be such that $U \neq 1$.
Then $\width{e} > 0$.
\end{restatable}

\subsection{Linear systems}\label{sec:linear-systems}

\KA is equipped to find the least solutions to linear inequations.
For instance, if we want to find $X$ such that $e \cdot X + f \leqqka X$, it is not hard to show that $e^\star \cdot f$ is the \emph{least solution} for $X$, in the sense that this choice of $X$ satisfies the inequation, and for any choice of $X$ that also satisfies this inequation it holds that $e^\star \cdot f \leqqka X$.
Since \KA is contained in \BKA and \CKA, the same constructions also apply there.
These axioms generalise to systems of linear inequations in a straightforward manner; indeed, Kozen~\cite{kozen-1994} exploited this generalisation to axiomatise \KA. 
In this paper, we use systems of linear inequations to construct particular expressions.
To do this, we introduce vectors and matrices of terms.

For the remainder of this section, we fix $I$ as a finite set.

\begin{definition}%
\label{definition:vector-matrix}
An \emph{$I$-vector} is a function from $I$ to $\terms$.
Addition of $I$-vectors is defined pointwise, i.e., if $p$ and $q$ are $I$-vectors, then $p + q$ is the $I$-vector defined for $i \in I$ by $(p + q)(i) \triangleq p(i) + q(i)$.

An \emph{$I$-matrix} is a function from $I^2$ to $\terms$.
Left-multiplication of an $I$-vector by an $I$-matrix is defined in the usual fashion, i.e., if $M$ is an $I$-matrix and $p$ is an $I$-vector, then $M \cdot p$ is the $I$-vector defined for $i \in I$ by
\[(M \cdot p)(i) \triangleq \sum_{j \in I} M(i, j) \cdot p(j)\]
\end{definition}

Equivalence between terms extends pointwise to $I$-vectors. More precisely, we write $p \equivaka q$ for $I$-vectors $p$ and $q$ when $p(i) \equivaka q(i)$ for all $i \in I$, and $p \leqqaka q$ when $p + q \equivaka q$.

\begin{definition}%
\label{definition:linear-system}
An \emph{$I$-linear system} $\ls{L}$ is a pair $\angl{M, p}$ where $M$ is an $I$-matrix and $p$ is an $I$-vector.
A \emph{solution} to $\ls{L}$ in \AKA is an $I$-vector $s$ such that $M \cdot s + p \leqqaka s$.
A \emph{least solution} to $\ls{L}$ in \AKA is a solution $s$ in \AKA such that for any solution $t$ in \AKA it holds that $s \leqqaka t$.
\end{definition}

It is not very hard to show that least solutions of a linear system are unique, up to $\equivaka$; we therefore speak of \emph{the} least solution of a linear system.

Interestingly, \emph{any} $I$-linear system has a least solution, and one can construct this solution using only the operators of \KA. 
The construction proceeds by induction on $|I|$.
In the base, where $I$ is empty, the solution is trivial; for the inductive step it suffices to reduce the problem to finding the least solution of a strictly smaller linear system.
This construction is not unlike Kleene's procedure to obtain a regular expression from a finite automaton~\cite{kleene-1956}.
Alternatively, we can regard the existence of least solutions as a special case of Kozen's proof of the fixpoint for matrices over a \KA, as seen in~\cite[Lemma 9]{kozen-1994}.

As a matter of fact, because this construction uses the axioms of \KA exclusively, the least solution that is constructed is the same for both \BKA and \CKA. 

\begin{restatable}{lemma}{linearsystemsolution}%
\label{lemma:linear-system-solution}
Let $\ls{L}$ be an $I$-linear system.
One can construct a single $I$-vector~$x$ that is the least solution to $\ls{L}$ in both \BKA and \CKA. 
\end{restatable}

\ifarxiv%
For the sake of self-containment, we include a full proof of the lemma above using the notation of this paper in \Cref{appendix:proofs-preliminaries}.
\else%
We include a full proof of the lemma above using the notation of this paper in the full version of this paper~\cite{brunet-kappe-silva-zanasi-2017-arxiv}.
\fi%

\section{Completeness of \texorpdfstring{\CKA}{CKA}}%
\label{section:axiomatising}

We now turn our attention to proving that $\equivcka$ is complete for \CKA-semantic equivalence of terms, i.e., that if $e, f \in \terms$ are such that $\semcka{e} = \semcka{f}$, then $e \equivcka f$.
In the interest of readability, proofs of technical lemmas in this section
\ifarxiv%
are deferred to \Cref{appendix:axiomatising-proofs}.
\else%
can be found in the full version of this paper~\cite{brunet-kappe-silva-zanasi-2017-arxiv}.
\fi%

As mentioned before, our proof of completeness is based on the completeness result for \BKA reproduced in~\Cref{theorem:wbka-completeness}.
Recall that $\semcka{e} = \down{\sembka{e}}$.
To reuse completeness of \BKA, we construct a syntactic variant of the closure operator, which is formalised below.
\begin{definition}%
\label{definition:closure}
Let $e \in \terms$.
We say that $\down{e}$ is a \emph{closure} of $e$ if both $e \equivcka \down{e}$ and $\sembka{\down{e}} = \down{\sembka{e}}$ hold.
\end{definition}

\begin{example}
Let $e = a \parallel b$; as proposed in \Cref{section:overview}, we claim that $\down{e} = a \parallel b + b \cdot a + a \cdot b$ is a closure of $e$.
To see why, first note that $e \leqqcka \down{e}$ by construction.
Furthermore,
\[
ab
    \equivcka (a \parallel 1) \cdot (1 \parallel b)
    \leqqcka (a \cdot 1) \parallel (1 \cdot b)
    \equivcka a \parallel b
\]
and similarly $ba \leqqcka e$; thus, $e \equivcka \down{e}$.
Lastly, the pomsets in $\down{\sembka{e}}$ and $\sembka{\down{e}}$ are simply $a \parallel b$, $ab$ and $ba$, and therefore $\sembka{\down{e}} = \down{\sembka{e}}$.
\end{example}

Laurence and Struth observed that the existence of a closure for every term implies a completeness theorem for \CKA, as follows.

\begin{lemma}%
\label{lemma:closure-implies-completeness}
Suppose that we can construct a closure for every element of $\terms$.
If $e, f \in \terms$ such that $\semcka{e} = \semcka{f}$, then $e \equivcka f$.
\end{lemma}
\begin{proof}
Since $\semcka{e} = \down{\sembka{e}} = \sembka{\down{e}}$ and similarly $\semcka{f} = \sembka{\down{f}}$, we have $\sembka{\down{e}} = \sembka{\down{f}}$.
By \Cref{theorem:wbka-completeness}, we get $\down{e} \equivbka \down{f}$, and thus $\down{e} \equivcka \down{f}$, since all axioms of \BKA are also axioms of \CKA. 
By $e \equivcka \down{e}$ and $\down{f} \equivcka f$, we can then conclude that $e \equivcka f$.
\end{proof}

The remainder of this section is dedicated to showing that the premise of \Cref{lemma:closure-implies-completeness} holds.
We do this by explicitly constructing a closure $\down{e}$ for every $e \in \terms$.
First, we note that closure can be constructed for the base terms.

\begin{restatable}{lemma}{closurebase}%
\label{lemma:closure-base}
Let $e \in 2$ or $e = a$ for some $a \in \alphabet$.
Then $e$ is a closure of itself.
\end{restatable}

Furthermore, closure can be constructed compositionally for all operators except parallel composition, in the following sense.

\begin{lemma}%
\label{lemma:closure-compositional}
Suppose that $e_0, e_1 \in \terms$, and that $e_0$ and $e_1$ have closures $\down{e_0}$ and $\down{e_1}$.
Then
\begin{inparaenum}[(i)] 
    \item $\down{e_0} + \down{e_1}$ is a closure of $e_0 + e_1$,
    \item $\down{e_0} \cdot \down{e_1}$ is a closure of $e_0 \cdot e_1$, and
    \item ${(\down{e_0})}^\star$ is a closure of $e_0^\star$.
\end{inparaenum}
\end{lemma}
\begin{proof}
Since $\down{e_0} \equivcka e_0$ and $\down{e_1} \equivcka e_1$, by the fact that $\equivcka$ is a congruence we obtain $\down{e_0} + \down{e_1} \equivcka e_0 + e_1$.
Similar observations hold for the other operators.
We conclude using \Cref{lemma:pl-composition-vs-closure}.
\end{proof}

It remains to consider the case where $e = e_0 \parallel e_1$.
In doing so, our induction hypothesis is that any $f \in \terms$ with $\width{f} < \width{e_0 \parallel e_1}$ has a closure, as well as any strict subterm of $e_0 \parallel e_1$.

\subsection{Preclosure}

To get to a closure of a parallel composition, we first need an operator on terms that is not a closure quite yet, but whose \BKA-semantics is ``closed enough'' to cover the non-sequential elements of the \CKA-semantics of the term.
\begin{definition}%
\label{definition:preclosure}
Let $e \in \terms$.
A \emph{preclosure} of $e$ is a term $\tilde{e} \in \terms$ such that $\tilde{e} \equivcka e$.
Moreover, if $U \in \semcka{e}$ is non-sequential, then $U \in \sembka{\tilde{e}}$.
\end{definition}

\begin{example}%
\label{example:preclosure}
Suppose that $e_0 \parallel e_1 = (a \parallel b) \parallel c$.
A preclosure of $e_0 \parallel e_1$ could be
\[\tilde{e} = a \parallel b \parallel c + (a \cdot b + b \cdot a) \parallel c + (b \cdot c + c \cdot b) \parallel a + (a \cdot c + c \cdot a) \parallel b\]
To verify this, note that $e \leqqcka \tilde{e}$ by construction; remains to show that $\tilde{e} \leqqcka e$.
This is fairly straightforward: since $a \cdot b + b \cdot a \leqqcka a \parallel b$, we have $(a \cdot b + b \cdot a) \parallel c \leqqcka e$; the other terms are treated similarly.
Consequently, $e \equivcka \tilde{e}$.
Furthermore, there are seven non-sequential pomsets in $\semcka{e}$; they are
\begin{mathpar}
a \parallel b \parallel c \and
ab \parallel c \and
ba \parallel c \and
bc \parallel a \and
cb \parallel a \and
ac \parallel b \and
ca \parallel b
\end{mathpar}
Each of these pomsets is found in $\sembka{\tilde{e}}$.
It should be noted that $\tilde{e}$ is \emph{not} a closure of $e$; to see this, consider for instance that $abc \in \semcka{e}$, while $abc \not\in \sembka{\tilde{e}}$.
\end{example}

The remainder of this section is dedicated to showing that, under the induction hypothesis, we can construct a preclosure for any parallelly composed term.
This is not perfectly straightforward; for instance, consider the term $e_0 \parallel e_1$ discussed in \Cref{example:preclosure}.
At first glance, one might be tempted to choose $\down{e_0} \parallel \down{e_1}$ as a preclosure, since $\down{e_0}$ and $\down{e_1}$ exist by the induction hypothesis.
In that case, $\down{e_0} = a \parallel b + a \cdot b + b \cdot a$ is a closure of $e_0$.
Furthermore, $\down{e_1} = c$ is a closure of $e_1$, by \Cref{lemma:closure-base}.
However, $\down{e_0} \parallel \down{e_1}$ is not a preclosure of $e_0 \parallel e_1$, since $(a \cdot c) \parallel b$ is non-sequential and found in $\semcka{e_0 \parallel e_1}$, but not in $\sembka{\down{e_0} \parallel \down{e_1}}$.

The problem is that the preclosure of $e_0$ and $e_1$ should also allow (partial) sequentialisation of \emph{parallel parts} of $e_0$ and $e_1$; in this case, we need to sequentialise the $a$ part of $a \parallel b$ with $c$, and leave $b$ untouched.
To do so, we need to be able to \emph{split} $e_0 \parallel e_1$ into pairs of constituent terms, each of which represents a possible way to divvy up its parallel parts.
For instance, we can split $e_0 \parallel e_1 = (a \parallel b) \parallel c$ parallelly into $a \parallel b$ and $c$, but also into $a$ and $b \parallel c$, or into $a \parallel c$ and $b$.
The definition below formalises this procedure.
\begin{definition}%
\label{definition:psplit}
Let $e \in \terms$; $\psplit{e}$ is the smallest relation on $\terms$ such that
\begin{mathpar}
\inferrule{~}{%
    1 \psplit{e} e
}
\and
\inferrule{~}{%
    e \psplit{e} 1
}
\and
\inferrule{%
    \lefts \psplit{e_0} \rights
}{%
    \lefts \psplit{e_1 + e_0} \rights
}
\and
\inferrule{%
    \lefts \psplit{e_1} \rights
}{%
    \lefts \psplit{e_0 + e_1} \rights
}
\and
\inferrule{%
    \lefts \psplit{e} \rights
}{%
    \lefts \psplit{e^\star} \rights
}
\\
\inferrule{%
    \lefts \psplit{e_0} \rights \\
    \epsilon(e_1) = 1
}{%
    \lefts \psplit{e_0 \cdot e_1} \rights
}
\and
\inferrule{%
    \lefts \psplit{e_1} \rights \\
    \epsilon(e_0) = 1
}{%
    \lefts \psplit{e_0 \cdot e_1} \rights
}
\and
\inferrule{%
    \lefts_0 \psplit{e_0} \rights_0 \\
    \lefts_1 \psplit{e_1} \rights_1
}{%
    \lefts_0 \parallel \lefts_1 \psplit{e_0 \parallel e_1} \rights_0 \parallel \rights_1
}
\end{mathpar}
\end{definition}

Given $e \in \terms$, we refer to $\psplit{e}$ as the \emph{parallel splitting relation} of $e$, and to the elements of $\psplit{e}$ as \emph{parallel splices} of $e$.
Before we can use $\psplit{e}$ to construct the preclosure of $e$, we go over a number of properties of the parallel splitting relation.
The first of these properties is that a given $e \in \terms$ has only finitely many parallel splices.
This will be useful later, when we involve \emph{all} parallel splices of $e$ in building a new term, i.e., to guarantee that the constructed term is finite.

\begin{restatable}{lemma}{psplitfinite}%
\label{lemma:psplit-finite}
For $e \in \terms$, $\psplit{e}$ is finite.
\end{restatable}

We furthermore note that the parallel composition of any parallel splice of $e$ is ordered below $e$ by $\leqqbka$.
This guarantees that parallel splices never contain extra information, i.e., that their semantics do not contain pomsets that do not occur in the semantics of $e$.
It also allows us to bound the width of the parallel splices by the width of the term being split, as a result of \Cref{lemma:width-vs-equivalence}.

\begin{restatable}{lemma}{psplitdomination}%
\label{lemma:psplit-domination}
Let $e \in \terms$.
If $\lefts \psplit{e} \rights$, then $\lefts \parallel \rights \leqqbka e$.
\end{restatable}

\begin{corollary}%
\label{corollary:psplit-width}
Let $e \in \terms$.
If $\lefts \psplit{e} \rights$, then $\width{\lefts} + \width{\rights} \leq \width{e}$.
\end{corollary}

Finally, we show that $\psplit{e}$ is \emph{dense} when it comes to parallel pomsets, meaning that if we have a parallelly composed pomset in the semantics of $e$, then we can find a parallel splice where one parallel component is contained in the semantics of one side of the pair, and the other component in that of the other.

\begin{lemma}%
\label{lemma:psplit-density}
Let $e \in \terms$, and let $V, W$ be pomsets such that $V \parallel W \in \sembka{e}$.
Then there exist $\lefts, \rights \in \terms$ with $\lefts \psplit{e} \rights$ such that $V \in \sembka{\lefts}$ and $W \in \sembka{\rights}$.
\end{lemma}
\begin{proof}
The proof proceeds by induction on $e$.
In the base, we can discount the case where $e = 0$, for then the claim holds vacuously.
This leaves us two cases.
\begin{itemize}
    \item
    If $e = 1$, then $V \parallel W \in \sembka{e}$ entails $V \parallel W = 1$.
    By \Cref{lemma:pomset-unique-decomposition}, we find that $V = W = 1$.
    Since $1 \psplit{e} 1$ by definition of $\psplit{e}$, the claim follows when we choose $\lefts = \rights = 1$.

    \item
    If $e = a$ for some $a \in \alphabet$, then $V \parallel W \in \sembka{e}$ entails $V \parallel W = a$.
    By \Cref{lemma:pomset-unique-decomposition}, we find that either $V = 1$ and $W = a$, or $V = a$ and $W = 1$.
    In the former case, we can choose $\lefts = 1$ and $\rights = a$, while in the latter case we can choose $\lefts = a$ and $\rights = 1$.
    It is then easy to see that our claim holds in either case.
\end{itemize}

\noindent
For the inductive step, there are four cases to consider.
\begin{itemize}
    \item
    If $e = e_0 + e_1$, then $U_0 \parallel U_1 \in \sembka{e_i}$ for some $i \in 2$.
    But then, by induction, we find $\lefts, \rights \in \terms$ with $\lefts \psplit{e_i} \rights$ such that $V \in \sembka{\lefts}$ and $W \in \sembka{\rights}$.
    Since this implies that $\lefts \psplit{e} \rights$, the claim follows.

    \item
    If $e = e_0 \cdot e_1$, then there exist pomsets $U_0, U_1$ such that $V \parallel W = U_0 \cdot U_1$, and $U_i \in \sembka{e_i}$ for all $i \in 2$.
    By \Cref{lemma:pomset-unique-decomposition}, there are two cases to consider.
    \begin{itemize}
        \item
        Suppose that $U_i = 1$ for some $i \in 2$, meaning that $V \parallel W = U_0 \cdot U_1 = U_{1-i} \in \sembka{e_{1-i}}$ for this $i$.
        By induction, we find $\lefts, \rights \in \terms$ with $\lefts \psplit{e_{1-i}} \rights$, and $V \in \sembka{\lefts}$ as well as $W \in \sembka{\rights}$.
        Since $U_i = 1 \in \sembka{e_i}$, we have that $\epsilon(e_i) = 1$ by \Cref{lemma:nullable}, and thus $\lefts \psplit{e} \rights$.

        \item
        Suppose that $V = 1$ or $W = 1$.
        In the former case, $V \parallel W = W = U_0 \cdot U_1 \in \semcka{e}$.
        We then choose $\lefts = 1$ and $\rights = e$ to satisfy the claim.
        In the latter case, we can choose $\lefts = e$ and $\rights = 1$ to satisfy the claim analogously.
    \end{itemize}

    \item
    If $e = e_0 \parallel e_1$, then there exist pomsets $U_0, U_1$ such that $V \parallel W = U_0 \parallel U_1$, and $U_i \in \sembka{e_i}$ for all $i \in 2$.
    By \Cref{lemma:pomset-levi-parallel}, we find pomsets $V_0, V_1, W_0, W_1$ such that $V = V_0 \parallel V_1$, $W = W_0 \parallel W_1$, and $U_i = V_i \parallel W_i$ for $i \in 2$.
    For $i \in 2$, we then find by induction $\lefts_i, \rights_i \in \terms$ with $\lefts_i \psplit{e_i} \rights_i$ such that $V_i \in \sembka{\lefts_i}$ and $W_i \in \sembka{\rights_i}$.
    We then choose $\lefts = \lefts_0 \parallel \lefts_1$ and $\rights = \rights_0 \parallel \rights_1$.
    Since $V = V_0 \parallel V_1$, it follows that $V \in \sembka{\lefts}$, and similarly we find that $W \in \sembka{\rights}$.
    Since $\lefts \psplit{e} \rights$, the claim follows.

    \item
    If $e = e_0^\star$, then there exist $U_0, U_1, \dots, U_{n-1} \in \sembka{e_0}$ such that $V \parallel W = U_0 \cdot U_1 \cdots U_{n-1}$.
    If $n = 0$, i.e., $V \parallel W = 1$, then $V = W = 1$.
    In that case, we can choose $\lefts = e$ and $\rights = 1$ to find that $\lefts \psplit{e} \rights$, $V \in \sembka{\lefts}$ and $W \in \sembka{\rights}$, satisfying the claim.

    If $n > 0$, we can assume without loss of generality that, for $0 \leq i < n$, it holds that $U_i \neq 1$.
    By \Cref{lemma:pomset-unique-decomposition}, there are two subcases to consider.
    \begin{itemize}
        \item
        Suppose that $V, W \neq 1$; then $n = 1$ (for otherwise $U_j = 1$ for some $0 \leq j < n$ by \Cref{lemma:pomset-unique-decomposition}, which contradicts the above).
        Since $V \parallel W = U_0 \in \sembka{e_0}$, we find by induction $\lefts, \rights \in \terms$ with $\lefts \psplit{e_0} \rights$ such that $V \in \sembka{\lefts}$ and $W \in \sembka{\rights}$.
        The claim then follows by the fact that $\lefts \psplit{e} \rights$.

        \item
        Suppose that $V = 1$ or $W = 1$.
        In the former case, $V \parallel W = W = U_0 \cdot U_1 \cdots U_{n-1} \in \semcka{e}$.
        We then choose $\lefts = 1$ and $\rights = e$ to satisfy the claim.
        In the latter case, we can choose $\lefts = e$ and $\rights = 1$ to satisfy the claim analogously.
        \qedhere
    \end{itemize}
\end{itemize}
\end{proof}

\begin{example}
Let $U = a \parallel c$ and $V = b$, and note that $U \parallel V \in \semcka{e_0 \parallel e_1}$.
We can then find that $a \psplit{a} 1$ and $1 \psplit{b} b$, and thus $a \parallel 1 \psplit{e_0} 1 \parallel b$.
Since also $c \psplit{c} 1$, it follows that $(a \parallel 1) \parallel c \psplit{e_0 \parallel e_1} (1 \parallel b) \parallel 1$.
We can then choose $\lefts = (a \parallel 1) \parallel c$ and $\rights = (1 \parallel b) \parallel 1$ to find that $U \in \sembka{\lefts}$ and $V \in \sembka{\rights}$, while $\lefts \psplit{e_0 \parallel e_1} \rights$.
\end{example}

With parallel splitting in hand, we can define an operator on terms that combines all parallel splices of a parallel composition in a way that accounts for all of their downward closures.
\begin{definition}
Let $e, f \in \terms$, and suppose that, for every $g \in \terms$ such that $\width{g} < \width{e} + \width{f}$, there exists a closure $\down{g}$.
The term $e \pc f$ is defined as follows:
\[e \pc f \triangleq e \parallel f + \sum_{\substack{\lefts \psplit{e \parallel f} \rights \\ \width{\lefts}, \width{\rights} < \width{e \parallel f}}} \down{\lefts} \parallel \down{\rights}\]
\end{definition}

Note that $e \pc f$ is well-defined: the sum is finite since $\psplit{e \parallel f}$ is finite by \Cref{lemma:psplit-finite}, and furthermore $\down{\lefts}$ and $\down{\rights}$ exist, as we required that $\width{\lefts}, \width{\rights} < \width{e \parallel f}$.

\begin{example}
Let us compute $e_0 \pc e_1$ and verify that we obtain a preclosure of $e_0 \parallel e_1$.
Working through the definition, we see that $\psplit{e_0 \parallel e_1}$ consists of the pairs
\begin{mathpar}
\angl{(1 \parallel 1) \parallel 1, (a \parallel b) \parallel c} \and
\angl{(1 \parallel 1) \parallel c, (a \parallel b) \parallel 1} \and
\angl{(1 \parallel b) \parallel 1, (a \parallel 1) \parallel c} \and
\angl{(1 \parallel b) \parallel c, (a \parallel 1) \parallel 1} \and
\angl{(a \parallel 1) \parallel 1, (1 \parallel b) \parallel c} \and
\angl{(a \parallel 1) \parallel c, (1 \parallel b) \parallel 1} \and
\end{mathpar}
Since closure is invariant with respect to $\equivcka$, we can simplify these terms by applying the axioms of $\CKA$.
After folding the unit subterms, we are left with
\begin{mathpar}
\angl{1, a \parallel b \parallel c} \and
\angl{c, a \parallel b} \and
\angl{b, a \parallel c} \and
\angl{b \parallel c, a} \and
\angl{a, b \parallel c} \and
\angl{a \parallel c, b}
\end{mathpar}
Recall that $a \parallel b + a \cdot b + b \cdot a$ is a closure of $a \parallel b$.
Now, we find that
\begin{align*}
e_0 \pc e_1
    &= (a \parallel b) \parallel c + c \parallel (a \parallel b + a \cdot b + b \cdot a) \\
    &\phantom{=} + b \parallel (a \parallel c + a \cdot c + c \cdot a) + (b \parallel c + b \cdot c + c \cdot b) \parallel a \\
    &\phantom{=} + a \parallel (b \parallel c + b \cdot c + c \cdot b) + (a \parallel c + a \cdot c + c \cdot a) \parallel b \\
    &\equivcka
    a \parallel b \parallel c +
    a \parallel (b \cdot c + c \cdot b) +
    b \parallel (a \cdot c + c \cdot a) +
    c \parallel (a \cdot b + b \cdot a)
\end{align*}
which was shown to be a preclosure of $e_0 \parallel e_1$ in \Cref{example:preclosure}.
\end{example}

The general proof of correctness for $\pc$ as a preclosure plays out as follows.
\begin{lemma}%
\label{lemma:preclosure}
Let $e, f \in \terms$, and suppose that, for every $g \in \terms$ with $\width{g} < \width{e} + \width{f}$, there exists a closure $\down{g}$.
Then $e \pc f$ is a preclosure of $e \parallel f$.
\end{lemma}
\begin{proof}
We start by showing that $e \pc f \equivcka e \parallel f$.
First, note that $e \parallel f \leqqbka e \pc f$ by definition of $e \pc f$.
For the other direction, suppose that $\lefts, \rights \in \terms$ are such that $\lefts \psplit{e \parallel f} \rights$.
By definition of closure, we know that $\down{\lefts} \parallel \down{\rights} \equivcka \lefts \parallel \rights$.
By \Cref{lemma:psplit-domination}, we have $\lefts \parallel \rights \leqqbka e \parallel f$.
Since every subterm of $e \pc f$ is ordered below $e \parallel f$ by $\leqqcka$, we have that $e \pc f \leqqcka e \parallel f$.
It then follows that $e \parallel f \equivcka e \pc f$.

For the second requirement, suppose that $X \in \semcka{e \parallel f}$ is non-sequential.
We then know that there exists a $Y \in \sembka{e \parallel f}$ such that $X \sqsubseteq Y$.
This leaves us two cases to consider.
\begin{itemize}
    \item
    If $X$ is empty or primitive, then $Y = X$ by \Cref{lemma:pomset-subsumption-base}, thus $X \in \sembka{e \parallel f}$.
    By the fact that $e \parallel f \leqqbka e \pc f$ and by \Cref{lemma:soundness}, we find $X \in \sembka{e \pc f}$.

    \item
    If $X = X_0 \parallel X_1$ for non-empty pomsets $X_0$ and $X_1$, then by \Cref{lemma:pomset-factorise-subsumption} we find non-empty pomsets $Y_0$ and $Y_1$ with $Y = Y_0 \parallel Y_1$ such that $X_i \sqsubseteq Y_i$ for $i \in 2$.
    By \Cref{lemma:psplit-density}, we find $\lefts, \rights \in \terms$ with $\lefts \psplit{e \parallel f} \rights$ such that $Y_0 \in \sembka{\lefts}$ and $Y_1 \in \sembka{\rights}$.
    By \Cref{lemma:width-vs-nonempty}, we find that $\width{\lefts}, \width{\rights} \geq 1$.
    \Cref{corollary:psplit-width} then allows us to conclude that $\width{\lefts}, \width{\rights} < \width{e \parallel f}$.

    This means that $\down{\lefts} \parallel \down{\rights} \leqqbka e \pc f$.
    Since $X_0 \in \sembka{\down{\lefts}}$ and $X_1 \in \sembka{\down{\rights}}$ by definition of closure, we can derive by \Cref{lemma:soundness} that
    \[X = X_0 \parallel X_1 \in \sembka{\down{\lefts} \parallel \down{\rights}} \subseteq \sembka{e \pc f} \qedhere\]
\end{itemize}
\end{proof}

\subsection{Closure}

The preclosure operator discussed above covers the non-sequential pomsets in the language $\semcka{e \parallel f}$; it remains to find a term that covers the sequential pomsets contained in $\semcka{e \parallel f}$.

To better give some intuition to the construction ahead, we first explore the observations that can be made when a sequential pomset $W \cdot X$ appears in the language $\semcka{e \parallel f}$; without loss of generality, assume that $W$ is non-sequential.
In this setting, there must exist $U \in \sembka{e}$ and $V \in \sembka{f}$ such that $W \cdot X \sqsubseteq U \parallel V$.
By \Cref{lemma:pomset-interpolation}, we find pomsets $U_0, U_1, V_0, V_1$ such that
\begin{mathpar}
W \sqsubseteq U_0 \parallel V_0 \and
X \sqsubseteq U_1 \parallel V_1 \and
U_0 \cdot U_1 \sqsubseteq U \and
V_0 \cdot V_1 \sqsubseteq V
\end{mathpar}
This means that $U_0 \cdot U_1 \in \semcka{e}$ and $V_0 \cdot V_1 \in \semcka{f}$.
Now, suppose we could find $e_0, e_1, f_0, f_1 \in \terms$ such that
\begin{mathpar}
e_0 \cdot e_1 \leqqcka e \and
U_0 \in \semcka{e_0} \and
U_1 \in \semcka{e_1} \\
f_0 \cdot f_1 \leqqcka f \and
V_0 \in \semcka{f_0} \and
V_1 \in \semcka{f_1}
\end{mathpar}
Then we have $W \in \sembka{e_0 \pc f_0}$, and $X \in \semcka{e_1 \parallel f_1}$.
Thus, if we can find a closure of $e_1 \parallel f_1$, then we have a term whose \BKA-semantics contains $W \cdot X$.

There are two obstacles that need to be resolved before we can use the observations above to find the closure of $e \parallel f$.
The first problem is that we need to be sure that this process of splitting terms into sequential components is at all possible, i.e., that we can split $e$ into $e_0$ and $e_1$ with $e_0 \cdot e_1 \leqqcka e$ and $U_i \in \semcka{e_i}$ for $i \in 2$.
We do this by designing a sequential analogue to the parallel splitting relation seen before.
The second problem, which we will address later in this section, is whether this process of splitting a parallel term $e \parallel f$ according to the exchange law and finding a closure of remaining term $e_1 \parallel f_1$ is well-founded, i.e., if we can find ``enough'' of these terms to cover all possible ways of sequentialising $e \parallel f$.
This will turn out to be possible, by using the fixpoint axioms of \KA as in~\Cref{sec:linear-systems} with linear systems.

We start by defining the sequential splitting relation.%
\footnote{The contents of this relation are very similar to the set of \emph{left- and right-spines} of a NetKAT expression as used in~\cite{foster-kozen-milano-silva-thompson-2015}.}
\begin{definition}%
\label{definition:ssplit}
Let $e \in \terms$; $\ssplit{e}$ is the smallest relation on $\terms$ such that
\begin{mathpar}
\inferrule{~}{%
    1 \ssplit{1} 1
}
\and
\inferrule{~}{%
    a \ssplit{a} 1
}
\and
\inferrule{~}{%
    1 \ssplit{a} a
}
\and
\inferrule{~}{%
    1 \ssplit{e_0^\star} 1
}
\and
\inferrule{%
    \lefts \ssplit{e_0} \rights
}{%
    \lefts \ssplit{e_0 + e_1} \rights
}
\and
\inferrule{%
    \lefts \ssplit{e_1} \rights
}{%
    \lefts \ssplit{e_0 + e_1} \rights
}
\\
\inferrule{%
    \lefts \ssplit{e_0} \rights
}{%
    \lefts \ssplit{e_0 \cdot e_1} \rights \cdot e_1
}
\and
\inferrule{%
    \lefts \ssplit{e_1} \rights
}{%
    e_0 \cdot \lefts \ssplit{e_0 \cdot e_1} \rights
}
\and
\inferrule{%
    \lefts_0 \ssplit{e_0} \rights_0 \\
    \lefts_1 \ssplit{e_1} \rights_1
}{%
    \lefts_0 \parallel \lefts_1 \ssplit{e_0 \parallel e_1} \rights_0 \parallel \rights_1
}
\and
\inferrule{%
    \lefts \ssplit{e_0} \rights
}{%
    e_0^\star \cdot \lefts \ssplit{e_0^\star} \rights \cdot e_0^\star
}
\end{mathpar}
\end{definition}

Given $e \in \terms$, we refer to $\ssplit{e}$ as the \emph{sequential splitting relation} of $e$, and to the elements of $\ssplit{e}$ as \emph{sequential splices} of $e$.
We need to establish a few properties of the sequential splitting relation that will be useful later on.
The first of these properties is that, as for parallel splitting, $\ssplit{e}$ is finite.
\begin{restatable}{lemma}{ssplitfinite}%
\label{lemma:ssplit-finite}
For $e \in \terms$, $\ssplit{e}$ is finite.
\end{restatable}

We also have that the sequential composition of splices is provably below the term being split.
Just like the analogous lemma for parallel splitting, this guarantees that our sequential splices never give rise to semantics not contained in the split term.
This lemma also yields an observation about the width of sequential splices when compared to the term being split.

\begin{restatable}{lemma}{ssplitdomination}%
\label{lemma:ssplit-domination}
Let $e \in \terms$.
If $\lefts, \rights \in \terms$ with $\lefts \ssplit{e} \rights$, then $\lefts \cdot \rights \leqqcka e$.
\end{restatable}

\begin{corollary}%
\label{corollary:ssplit-width}
Let $e \in \terms$.
If $\lefts, \rights \in \terms$ with $\lefts \ssplit{e} \rights$, then $\width{\lefts}, \width{\rights} \leq \width{e}$.
\end{corollary}

Lastly, we show that the splices cover every way of (sequentially) splitting up the semantics of the term being split, i.e., that $\ssplit{e}$ is dense when it comes to sequentially composed pomsets.

\begin{lemma}%
\label{lemma:ssplit-density}
Let $e \in \terms$, and let $V$ and $W$ be pomsets such that $V \cdot W \in \semcka{e}$.
Then there exist $\lefts, \rights \in \terms$ with $\lefts \ssplit{e} \rights$ such that $V \in \semcka{\lefts}$ and $W \in \semcka{\rights}$.
\end{lemma}
\begin{proof}
The proof proceeds by induction on $e$.
In the base, we can discount the case where $e = 0$, for then the claim holds vacuously.
This leaves us two cases.
\begin{itemize}
    \item
    If $e = 1$, then $V \cdot W = 1$; by \Cref{lemma:pomset-unique-decomposition}, we find that $V = W = 1$.
    Since $1 \ssplit{e} 1$ by definition of $\ssplit{e}$, the claim follows when we choose $\lefts = \rights = 1$.

    \item
    If $e = a$ for some $a \in \alphabet$, then $V \cdot W = a$; by \Cref{lemma:pomset-unique-decomposition}, we find that either $V = a$ and $W = 1$ or $V = 1$ and $W = a$.
    In the former case, we can choose $\lefts = a$ and $\rights = 1$ to satisfy the claim; the latter case can be treated similarly.
\end{itemize}

\noindent
For the inductive step, there are four cases to consider.
\begin{itemize}
    \item
    If $e = e_0 + e_1$, then $V \cdot W \in \semcka{e_i}$ for some $i \in 2$.
    By induction, we find $\lefts, \rights \in \terms$ with $\lefts \ssplit{e_i} \rights$ such that $V \in \semcka{\lefts}$ and $W \in \semcka{\rights}$.
    Since $\lefts \ssplit{e} \rights$ in this case, the claim follows.

    \item
    If $e = e_0 \cdot e_1$, then there exist $U_0 \in \semcka{e_0}$ and $U_1 \in \semcka{e_1}$ such that $V \cdot W = U_0 \cdot U_1$.
    By \Cref{lemma:pomset-levi-generalised}, we find a series-parallel pomset $X$ such that either $V \sqsubseteq U_0 \cdot X$ and $X \cdot W \sqsubseteq U_1$, or $V \cdot X \sqsubseteq U_0$ and $W \sqsubseteq X \cdot U_1$.
    In the former case, we find that $X \cdot W \in \semcka{e_1}$, and thus by induction $\lefts', \rights \in \terms$ with $\lefts' \ssplit{e_1} \rights$ such that $X \in \semcka{\lefts'}$ and $W \in \semcka{\rights}$.
    We then choose $\lefts = e_0 \cdot \lefts'$ to find that $\lefts \ssplit{e} \rights$, as well as $V \sqsubseteq U_0 \cdot X \in \semcka{e_0} \cdot \semcka{\lefts'} = \semcka{\lefts}$ and thus $V \in \semcka{\lefts}$.
    The latter case can be treated similarly; here, we use the induction hypothesis on $e_0$.

    \item
    If $e = e_0 \parallel e_1$, then there exist $U_0 \in \semcka{e_0}$ and $U_1 \in \semcka{e_1}$ such that $V \cdot W \sqsubseteq U_0 \parallel U_1$.
    By \Cref{lemma:pomset-interpolation}, we find series-parallel pomsets $V_0, V_1, W_0, W_1$ such that $V \sqsubseteq V_0 \parallel V_1$ and $W \sqsubseteq W_0 \parallel W_1$, as well as $V_i \cdot W_i \sqsubseteq U_i$ for all $i \in 2$.
    In that case, $V_i \cdot W_i \in \semcka{e_i}$ for all $i \in 2$, and thus by induction we find $\lefts_i, \rights_i \in \terms$ with $\lefts_i \ssplit{e_i} \rights_i$ such that $V_i \in \semcka{\lefts_i}$ and $W_i \in \semcka{\rights_i}$.
    We choose $\lefts = \lefts_0 \parallel \lefts_1$ and $\rights = \rights_0 \parallel \rights_1$ to find that $V \in \semcka{\lefts_0 \parallel \rights_0}$ and $W \in \semcka{\lefts_1 \parallel \rights_1}$, as well as $\lefts \ssplit{e} \rights$.

    \item
    If $e = e_0^\star$, then there exist $U_0, U_1, \dots, U_{n-1} \in \semcka{e_0}$ such that $V \cdot W = U_0 \cdot U_1 \cdots U_{n-1}$.
    Without loss of generality, we can assume that for $0 \leq i < n$ it holds that $U_i \neq 1$.
    In the case where $n = 0$ we have that $V \cdot W = 1$, thus $V = W = 1$, we can choose $\lefts = \rights = 1$ to satisfy the claim.

    For the case where $n > 0$, we find by \Cref{lemma:pomset-levi-generalised} an $0 \leq m < n$ and series-parallel pomsets $X, Y$ such that $X \cdot Y \sqsubseteq U_m$, and $V \sqsubseteq U_0 \cdot U_1 \cdots U_{m-1} \cdot X$ and $W \sqsubseteq Y \cdot U_{m+1} \cdot U_{m+2} \cdots U_n$.
    Since $X \cdot Y \sqsubseteq U_m \in \semcka{e_0}$ and thus $X \cdot Y \in \semcka{e_0}$, we find by induction $\lefts', \rights' \in \terms$ with $\lefts' \ssplit{e_0} \rights'$ and $X \in \semcka{\lefts'}$ and $Y \in \semcka{\rights'}$.
    We can then choose $\lefts = e_0^\star \cdot \lefts'$ and $\rights = \rights' \cdot e_0^\star$ to find that $V \sqsubseteq U_0 \cdot U_1 \cdots U_{m-1} \cdot X \in \semcka{e_0^\star} \cdot \semcka{\lefts'} = \semcka{\lefts}$ and $W \sqsubseteq Y \cdot U_{m+1} \cdot U_{m+2} \cdots U_n \in \semcka{\rights'} \cdot \semcka{e_0^\star} = \semcka{\rights}$, and thus that $V \in \semcka{\lefts}$ and $W \in \semcka{\rights}$.
    Since $\lefts \ssplit{e} \rights$ holds, the claim follows.
    \qedhere
\end{itemize}
\end{proof}

\begin{example}
Let $U$ be the pomset $ca$ and let $V$ be $bc$.
Furthermore, let $e$ be the term ${(a \cdot b + c)}^\star$, and note that $U \cdot V \in \semcka{e}$.
We then find that $a \ssplit{a} 1$, and thus $a \ssplit{a \cdot b} 1 \cdot b$.
We can now choose $\lefts = {(a \cdot b + c)}^\star \cdot a$ and $\rights = (1 \cdot b) \cdot {(a \cdot b + c)}^\star$ to find that $U \in \semcka{\lefts}$ and $V \in \semcka{\rights}$, while $\lefts \ssplit{e} \rights$.
\end{example}

We know how to split a term sequentially.
To resolve the second problem, we need to show that the process of splitting terms repeatedly ends somewhere.
This is formalised in the notion of \emph{right-hand remainders}, which are the terms that can appear as the right hand of a sequential splice of a term.

\begin{definition}%
\label{definition:remainders}
Let $e \in \terms$.
The set of \emph{(right-hand) remainders} of $e$, written $R(e)$, is the smallest satisfying the rules
\begin{mathpar}
\inferrule{~}{%
    e \in R(e)
}
\and
\inferrule{%
    f \in R(e) \\
    \lefts \ssplit{f} \rights
}{%
    \rights \in R(e)
}
\end{mathpar}
\end{definition}
\vspace{-0.5em} 
\begin{restatable}{lemma}{remaindersfinite}%
\label{lemma:remainders-finite}
Let $e \in \terms$.
$R(e)$ is finite.
\end{restatable}

With splitting and remainders we are in a position to define the linear system that will yield the closure of a parallel composition.
Intuitively, we can think of this system as an automaton: every variable corresponds to a state, and every row of the matrix describes the ``transitions'' of the corresponding state, while every element of the vector describes the language ``accepted'' by that state without taking a single transition.
Solving the system for a least fixpoint can be thought of as finding an expression that describes the language of the automaton.

\begin{definition}%
\label{definition:closure-system}
Let $e, f \in \terms$, and suppose that, for every $g \in \terms$ such that $\width{g} < \width{e} + \width{f}$, there exists a closure $\down{g}$.
We choose
\[I_{e,f} = \{ g \parallel h : g \in R(e), h \in R(f) \}\]
The $I_{e,f}$-vector $p_{e,f}$ and $I_{e,f}$-matrix $M_{e,f}$ are chosen as follows.
\begin{mathpar}
p_{e,f}(g \parallel h) \triangleq g \parallel f
\and
M_{e,f}(g \parallel h, g' \parallel h') \triangleq \sum_{\substack{\lefts_g \ssplit{g} g' \\ \lefts_h \ssplit{h} h'}} \lefts_g \pc \lefts_h
\end{mathpar}
$I_{e,f}$ is finite by \Cref{lemma:remainders-finite}.
We write $\ls{L}_{e,f}$ for the $I_{e,f}$-linear system $\angl{M_{e,f}, p_{e,f}}$.
\end{definition}

We can check that $M_{e,f}$ is well-defined.
First, the sum is finite, because $\ssplit{g}$ and $\ssplit{h}$ are finite by \Cref{lemma:ssplit-finite}.
Second, if $g \parallel h \in I$ and $\lefts_g, \rights_g, \lefts_h, \rights_h \in \terms$ such that $\lefts_g \ssplit{g} \rights_g$ and $\lefts_h \ssplit{h} \rights_h$, then $\width{\lefts_g} \leq \width{g} \leq \width{e}$ and $\width{\lefts_h} \leq \width{h} \leq \width{f}$ by \Cref{corollary:ssplit-width}, and thus, if $d \in \terms$ such that $\width{d} < \width{\lefts_g} + \width{\lefts_h}$, then $\width{d} < \width{e} + \width{f}$, and therefore a closure of $d$ exists, meaning that $\lefts_g \pc \lefts_h$ exists, too.

The least solution to $\ls{L}_{e,f}$ obtained through \Cref{lemma:linear-system-solution} is the $I$-vector denoted by $s_{e,f}$.
We write $e \ic f$ for $s_{e,f}(e \parallel f)$, i.e., the least solution at $e \parallel f$.

Using the previous lemmas, we can then show that $e \ic f$ is indeed a closure of $e \parallel f$, provided that we have closures for all terms of strictly lower width.
The intuition of this proof is that we use the uniqueness of least fixpoints to show that $e \parallel f \equivcka e \ic f$, and then use the properties of preclosure and the normal form of series-parallel pomsets to show that $\semcka{e \parallel f} = \sembka{e \ic f}$.

\begin{lemma}%
\label{lemma:closure-system}
Let $e, f \in \terms$, and suppose that, for every $g \in \terms$ with $\width{g} < \width{e} + \width{f}$, there exists a closure $\down{g}$.
Then $e \ic f$ is a closure of $e \parallel f$.
\end{lemma}
\begin{proof}
We begin by showing that $e \parallel f \equivcka e \ic f$.
We can see that $p_{e,f}$ is a solution to $\ls{L}_{e,f}$, by calculating for $g \parallel h \in I_{e,f}$:
\begin{align*}
    &(p_{e,f} + M_{e,f} \cdot p_{e,f})(g \parallel h)\\
    &= g \parallel h + \sum_{\rights_g \parallel \rights_h \in I} \Bigl(\sum_{\substack{\lefts_g \ssplit{g} \rights_g \\ \lefts_h \ssplit{h} \rights_h}} \lefts_g \pc \lefts_h\Bigr) \cdot (\rights_g \parallel \rights_h)
        \tag{def. $M_{e,f}$, $p_{e,f}$} \\
    &\equivcka g \parallel h + \sum_{\rights_g \parallel \rights_h \in I} \sum_{\substack{\lefts_g \ssplit{g} \rights_g \\ \lefts_h \ssplit{h} \rights_h}} (\lefts_g \pc \lefts_h) \cdot (\rights_g \parallel \rights_h)
        \tag{distributivity} \\
    &\equivcka g \parallel h + \sum_{\rights_g \parallel \rights_h \in I} \sum_{\substack{\lefts_g \ssplit{g} \rights_g \\ \lefts_h \ssplit{h} \rights_h}} (\lefts_g \parallel \lefts_h) \cdot (\rights_g \parallel \rights_h)
        \tag{\Cref{lemma:preclosure}} \\
    &\leqqcka g \parallel h + \sum_{\rights_g \parallel \rights_h \in I} \sum_{\substack{\lefts_g \ssplit{g} \rights_g \\ \lefts_h \ssplit{h} \rights_h}} (\lefts_g \cdot \rights_g) \parallel (\lefts_h \cdot \rights_h)
        \tag{exchange} \\
    &\leqqcka g \parallel h + \sum_{\rights_g \parallel \rights_h \in I} \sum_{\substack{\lefts_g \ssplit{g} \rights_g \\ \lefts_h \ssplit{h} \rights_h}} g \parallel h
        \tag{\Cref{lemma:ssplit-domination}} \\
    &\equivcka g \parallel h
        \tag{idempotence} \\
    &= p_{e,f}(g \parallel h)
        \tag{def. $p_{e,f}$}
\end{align*}
To see that $p_{e,f}$ is the \emph{least} solution to $\ls{L}_{e,f}$, let $q_{e,f}$ be a solution to $\ls{L}_{e,f}$.
We then know that $M_{e,f} \cdot q_{e,f} + p_{e,f} \leqqcka q_{e,f}$; thus, in particular, $p_{e,f} \leqqcka q_{e,f}$.
Since the least solution to a linear system is unique up to $\equivcka$, we find that $s_{e,f} \equivcka p_{e,f}$, and therefore that $e \ic f = s_{e,f}(e \parallel f) \equivcka p_{e,f}(e \parallel f) = e \parallel f$.

It remains to show that if $U \in \semcka{e \parallel f}$, then $U \in \sembka{e \ic f}$.
To show this, we show the more general claim that if $g \parallel h \in I$ and $U \in \semcka{g \parallel h}$, then $U \in \sembka{s_{e,f}(g \parallel h)}$.
Write $U = U_0 \cdot U_1 \cdots U_{n-1}$ such that for $0 \leq i < n$, $U_i$ is non-sequential (as in \Cref{lemma:pomset-normal-form}).
The proof proceeds by induction on~$n$.
In the base, we have that $n = 0$.
In this case, $U = 1$, and thus $U \in \sembka{g \parallel h}$ by \Cref{lemma:pomset-subsumption-base}.
Since $g \parallel h = p_{e,f}(g \parallel h) \leqqbka s_{e,f}(g \parallel h)$, it follows that $U \in \sembka{s_{e,f}(g \parallel h)}$ by \Cref{lemma:soundness}.

For the inductive step, assume the claim holds for $n-1$.
We write $U = U_0 \cdot U'$, with $U' = U_1 \cdot U_2 \cdots U_{n-1}$.
Since $U_0 \cdot U' \in \semcka{g \parallel h}$, there exist $W \in \semcka{g}$ and $X \in \semcka{h}$ such that $U_0 \cdot U' \sqsubseteq W \parallel X$.
By \Cref{lemma:pomset-interpolation}, we find pomsets $W_0, W_1, X_0, X_1$ such that $W_0 \cdot W_1 \sqsubseteq W$ and $X_0 \cdot X_1 \sqsubseteq X$, as well as $U_0 \sqsubseteq W_0 \parallel X_0$ and $U' \sqsubseteq W_1 \parallel X_1$.
By \Cref{lemma:ssplit-density}, we find $\lefts_g, \rights_g, \lefts_h, \rights_h \in \terms$ with $\lefts_g \ssplit{g} \rights_g$ and $\lefts_h \ssplit{h} \rights_h$, such that $W_0 \in \semcka{\lefts_g}$, $W_1 \in \semcka{\rights_g}$, $X_0 \in \semcka{\lefts_h}$ and $X_1 \in \semcka{\rights_h}$.

From this, we know that $U_0 \in \semcka{\lefts_g \parallel \lefts_h}$ and $U' \in \semcka{\rights_g \parallel \rights_h}$.
Since $U_0$ is non-sequential, we have that $U_0 \in \sembka{\lefts_g \pc \lefts_h}$.
Moreover, by induction we find that $U' \in \sembka{s_{e,f}(\rights_g \parallel \rights_h)}$.
Since $\lefts_g \pc \lefts_h \leqqbka M_{e,f}(g \parallel h, \rights_g \parallel \rights_h)$ by definition of $M_{e,f}$, we furthermore find that
\[ (\lefts_g \pc \lefts_h) \cdot s_{e,f}(\rights_g \parallel \rights_h) \leqqbka M_{e,f}(g \parallel h, \rights_g \parallel \rights_h) \cdot s_{e,f}(\rights_g \parallel \rights_h) \]
Since $\rights_g \parallel \rights_h \in I$, we find by definition of the solution to a linear system that
\[ M_{e,f}(g \parallel h, \rights_g \parallel \rights_h) \cdot s_{e,f}(\rights_g \parallel \rights_h) \leqqbka s_{e,f}(g \parallel h) \]
By \Cref{lemma:soundness} and the above, we conclude that $U = U_0 \cdot U' \in \sembka{s_{e,f}(g \parallel h)}$.
\end{proof}

For a concrete example where we find a closure of a (non-trivial) parallel composition by solving a linear system, we refer to \Cref{appendix:example-solve}.

With closure of parallel composition, we can construct a closure for any term and therefore conclude completeness of \CKA. 

\begin{theorem}%
\label{theorem:closure}
Let $e \in \terms$.
We can construct a closure $\down{e}$ of $e$.
\end{theorem}
\begin{proof}
The proof proceeds by induction on $\width{e}$ and the structure of $e$, i.e., by considering $f$ before $g$ if $\width{f} < \width{g}$, or if $f$ is a strict subterm of $g$ (in which case $\width{f} \leq \width{g}$ also holds).
It is not hard to see that this induces a well-ordering on $\terms$.

Let $e$ be a term of width $n$, and suppose that the claim holds for all terms of width at most $n-1$, and for all strict subterms of $e$.
There are three cases.
\begin{itemize}
    \item
    If $e = 0$, $e = 1$ or $e = a$ for some $a \in \alphabet$, the claim follows from \Cref{lemma:closure-base}.

    \item
    If $e = e_0 + e_1$, or $e = e_0 \cdot e_1$, or $e = e_0^\star$, the claim follows from \Cref{lemma:closure-compositional}.

    \item
    If $e = e_0 \parallel e_1$, then $e_0 \ic e_1$ exists by the induction hypothesis.
    By \Cref{lemma:closure-system}, we then find that $e_0 \ic e_1$ is a closure of $e$.
    \qedhere
\end{itemize}
\end{proof}

\begin{corollary}
Let $e, f \in \terms$.
If $\semcka{e} = \semcka{f}$, then $e \equivcka f$.
\end{corollary}
\begin{proof}
Follows from \Cref{theorem:closure} and \Cref{lemma:closure-implies-completeness}.
\end{proof}

\section{Discussion and further work}%
\label{section:discussion-further-work}

By building a syntactic closure for each series-rational expression, we have shown that the standard axiomatisation of \CKA is complete with respect to the \CKA-semantics of series-rational terms.
Consequently, the algebra of closed series-rational pomset languages forms the free \CKA. 

Our result leads to several decision procedures for the equational theory of \CKA. 
For instance, one can compute the closure of a term as described in the present paper, and use an existing decision procedure for \BKA~\cite{jategaonkar-meyer-1996,laurence-struth-2014,brunet-pous-struth-2017}.
Note however that although this approach seems suited for theoretical developments (such as formalising the results in a proof assistant), its complexity makes it less appealing for practical use.
More practically, one could leverage recent work by Brunet, Pous and Struth~\cite{brunet-pous-struth-2017}, which provides an algorithm to compare closed series-rational pomset languages.
Since this is the free concurrent Kleene algebra, this algorithm can now be used to decide the equational theory of \CKA. 
We also obtain from the latter paper that this decision problem is \textsc{expspace}-complete.

We furthermore note that the algorithm to compute downward closure can be used to extend half of the result from~\cite{kappe-brunet-luttik-silva-zanasi-2017} to a Kleene theorem that relates the \CKA-semantics of expressions to the pomset automata proposed there: if $e \in \terms$, we can construct a pomset automaton $A$ with a state $q$ such that $L_A(q) = \semcka{e}$.

Having established pomset automata as an operational model of \CKA, a further question is whether these automata are amenable to a bisimulation-based equivalence algorithm, as is the case for finite automata~\cite{hopcroft-karp-1971}.
If this is the case, optimisations such as those in~\cite{bonchi-pous-2013} might have analogues for pomset automata that can be found using the coalgebraic method~\cite{rot-bonsangue-rutten-2013}.

While this work was in development, an unpublished draft by Laurence and Struth~\slpaper{} appeared, with a first proof of completeness for \CKA. 
The general outline of their proof is similar to our own, in that they prove that closure of pomset languages preserves series-rationality, and hence there exists a syntactic closure for every series-rational expression.
However, the techniques used to establish this fact are quite different from the developments in the present paper.
First, we build the closure via syntactic methods: explicit splitting relations and solutions of linear systems.
Instead, their proof uses automata theoretic constructions and algebraic closure properties of regular languages; in particular, they rely on congruences of finite index and language homomorphisms.
We believe that our approach leads to a substantially simpler and more transparent proof.
Furthermore, even though Laurence and Struth do not seem to use any fundamentally non-constructive argument, their proof does not obviously yield an algorithm to effectively compute the closure of a given term.
In contrast, our proof is explicit enough to be implemented directly; we wrote a simple Python script (under six hundred lines) to do just that~\cite{brunet-kappe-silva-zanasi-2017-closure-script}.

A crucial ingredient in this work was the computation of least solutions of linear systems.
This kind of construction has been used on several occasions for the study of Kleene algebras~\cite{conway-1971,backhouse-1975,kozen-1994}, and we provide here yet another variation of such a result.
We feel that linear systems may not have yet been used to their full potential in this context, and could still lead to interesting developments.

A natural extension of the work conducted here would be to turn our attention to the signature of concurrent Kleene algebra that includes a ``parallel star'' operator $e^\parallel$.
The completeness result of Laurence and Struth~\cite{laurence-struth-2014} holds for \BKA with the parallel star, so in principle one could hope to extend our syntactic closure construction to include this operator.
Unfortunately, using the results of Laurence and Struth, we can show that this is not possible.
They defined a notion of \emph{depth} of a series-parallel pomset, intuitively corresponding to the nesting of parallel and sequential components.
An important step in their development consists of proving that for every series-parallel-rational language there exists a finite upper bound on the depth of its elements.
However, the language $\semcka{a^\parallel}$ does not enjoy this property: it contains every series-parallel pomset exclusively labelled with the symbol $a$.
Since we can build such pomsets with arbitrary depth, it follows that there does not exist a syntactic closure of the term $a^\parallel$.
New methods would thus be required to tackle the parallel star operator.

Another aspect of \CKA that is not yet developed to the extent of \KA is the coalgebraic perspective.
We intend to investigate whether the coalgebraic tools developed for \KA can be extended to \CKA, which will hopefully lead to efficient bisimulation-based decision procedures~\cite{bonchi-pous-2013,foster-kozen-milano-silva-thompson-2015}.

\paragraph{Acknowledgements}
We thank the anonymous reviewers for their insightful comments.
This work was partially supported by the ERC Starting Grant ProFoundNet (grant code 679127).

\appendix

\ifarxiv%

\section{Proofs for \Cref{section:preliminaries}}%
\label{appendix:proofs-preliminaries}

The notion of \N-freeness for pomsets is useful for proving the lemmas to come.
\begin{definition}
Let $U = [\lp{u}]$ be a pomset.
We say that $U$ is \emph{\N-free} if there are no $u_0, u_1, u_2, u_3 \in S_\lp{u}$ such that $u_0 \leq_\lp{u} u_1$, $u_2 \leq_\lp{u} u_3$ and $u_0 \leq_\lp{u} u_3$ and no other relation between them, i.e., the graph of these elements has the shape of an \N.
\end{definition}
Note that \N-freeness is well-defined for pomsets, for the presence of an \N-shape does not depend on the particular representative $\lp{u}$.
It is not hard to see that all series-parallel pomsets are \N-free.
Perhaps surprisingly, this \N-freeness provides a complete characterisation of series-parallel pomsets~\cite{gischer-1988}.

\begin{lemma}[Gischer]%
\label{lemma:pomset-sp-nfree}
A pomset is series-parallel if and only if it is \N-free.
\end{lemma}

It is also useful to restrict a labelled poset to a part of its carrier, as follows.
\begin{definition}%
\label{definition:lp-restriction}
Let $\lp{u}$ be a labelled poset, and let $S \subseteq S_\lp{u}$.
We write $\lp{u} \restr{S}$ for the \emph{restriction} of $\lp{u}$ to $S$, i.e., labelled poset given by $S_{\lp{u} \restr{S}} = S$, $\leq_{\lp{u} \restr{S}} =\ \leq_{\lp{u}} \cap\ S \times S$, and $\lambda_{\lp{u} \restr{S}}(z) = \lambda_\lp{u}(z)$.
\end{definition}

\subsection{Subsumption of empty or primitive pomsets}

\begin{lemma}%
\label{lemma:lp-subsumption-unit}
Let $\lp{u}$ be a labelled poset such that $\lp{u} \sqsubseteq \lp{1}$ or $\lp{1} \sqsubseteq \lp{u}$.
Then $\lp{u} = \lp{1}$.
\end{lemma}
\begin{proof}
We treat the case where $\lp{u} \sqsubseteq \lp{1}$; the case where $\lp{1} \sqsubseteq \lp{u}$ is similar.
Let $h: \lp{1} \to \lp{u}$ witness that $\lp{u} \sqsubseteq \lp{1}$.
Then $h$ is a bijection from $S_\lp{1} = \emptyset$ to $S_\lp{u}$; accordingly, $S_\lp{u} = \emptyset$.
But then $\lp{u} = \lp{1}$, because the labelled poset with empty carrier is unique.
\end{proof}

\begin{lemma}%
\label{lemma:lp-subsumption-primitive}
Let $\lp{u}, \lp{v}$ be a labelled posets, with $S_\lp{v}$ a singleton, such that $\lp{u} \sqsubseteq \lp{v}$ or $\lp{v} \sqsubseteq \lp{u}$.
Then $\lp{u} \simeq \lp{v}$.
\end{lemma}
\begin{proof}
We treat the case where $\lp{u} \sqsubseteq \lp{v}$; the case where $\lp{v} \sqsubseteq \lp{u}$ is similar.
Let $h: \lp{v} \to \lp{u}$ witness that $\lp{u} \sqsubseteq \lp{v}$.
Then $h$ is a bijection from $S_\lp{v}$ to $S_\lp{u}$; consequently, $S_\lp{u}$ is a singleton.
Now, if $u \leq_\lp{u} u'$, then $u, u' \in S_\lp{u}$ and thus $u = u'$.
Consequently, $h^{-1}(u) = h^{-1}(u')$, and thus $h^{-1}(u) \leq_\lp{v} h^{-1}(u')$.
Since furthermore $\lambda_\lp{u} \circ h = \lambda_\lp{v}$, also $\lambda_\lp{v} = \lambda_\lp{v} \circ h^{-1}$.
It follows that $h^{-1}: \lp{u} \to \lp{v}$ is a subsumption witnessing that $\lp{v} \sqsubseteq \lp{u}$.
We can thus conclude that $\lp{u} \simeq \lp{v}$.
\end{proof}

\pomsetsubsumptionbase*
\begin{proof}
First, suppose that $U = 1$.
We then have that $U = [\lp{1}]$ and $V = [\lp{v}]$ such that $\lp{u} \sqsubseteq \lp{1}$ or $\lp{1} \sqsubseteq \lp{u}$.
By \Cref{lemma:lp-subsumption-unit}, we find that $\lp{v} = \lp{1}$ and thus $V = [\lp{1}] = 1$.

Second, suppose that $U = a$ for some $a \in \alphabet$.
Then $U = [\lp{u}]$ for some pomset with singleton carrier $S_\lp{u}$, with $\lambda_\lp{u}(u) = a$ for all $u \in S_\lp{u}$.
Since $V = [\lp{v}]$ and $\lp{u} \sqsubseteq \lp{v}$ or $\lp{v} \sqsubseteq \lp{u}$, we find that $\lp{u} \simeq \lp{v}$ by \Cref{lemma:lp-subsumption-primitive}.
This establishes that $U = V$.
\end{proof}

\subsection{The factorisation lemma}

\pomsetfactorisesubsumption*
\begin{proof}
We start with the first claim.
Let $U$, $V_0$ and $V_1$ be as in the premise, and write $U = [\lp{u}]$, $V_0 = [\lp{v}_0]$ and $V_1 = [\lp{v}_1]$.
Without loss of generality, we can assume that $\lp{v}_0$ and $\lp{v}_1$ are disjoint, that $S_{\lp{v}_0} \cup S_{\lp{v}_1} = S_\lp{u}$, and that the identity function $S_{\lp{v}_0} \cup S_{\lp{v}_1} \to S_\lp{u}$ is the subsumption witnessing that $\lp{u} \sqsubseteq \lp{v}_0 \cdot \lp{v}_1$.

We then choose $\lp{u}_i = \lp{u} \restr{\lp{v}_i}$ for $i \in 2$, and claim that $\lp{u}_0 \cdot \lp{u}_1 = \lp{u}$.
\begin{itemize}
    \item
    For the carrier, we already know that
    \[S_{\lp{u}_0 \cdot \lp{u}_1} = S_{\lp{u}_0} \cup S_{\lp{u}_1} = (S_\lp{u} \cap S_{\lp{v}_0}) \cup (S_\lp{u} \cap S_{\lp{v}_1}) = S_\lp{u} \cap (S_{\lp{v}_0} \cup S_{\lp{v}_1}) = S_\lp{u}\]

    \item
    Now suppose that $u, u' \in S_\lp{u}$ such that $u \leq_{\lp{u}_0 \cdot \lp{u}_1} u'$.
    There are two cases:
    \begin{itemize}
        \item
        If $u, u' \in S_{\lp{v}_i}$ for some $i \in 2$, then $u \leq_{\lp{v}_i} u'$, and thus $u \leq_{\lp{v}_0 \cdot \lp{v}_1} u'$, meaning that $u \leq_\lp{u} u'$.

        \item
        If $u \in S_{\lp{v}_0}$ and $u' \in S_{\lp{v}_1}$, then $u \leq_{\lp{v}_0 \cdot \lp{v}_1} u'$, and thus $u \leq_\lp{u} u'$.
    \end{itemize}

    In the other direction, let $u, u' \in S_\lp{u}$ with $u \leq_\lp{u} u'$.
    There are three cases.
    \begin{itemize}
        \item
        If $u, u' \in S_{\lp{v}_i}$ for some $i \in 2$, then $u \leq_{\lp{v}_i} u'$, and thus $u \leq_{\lp{u}_i} u'$ and therefore $u \leq_{\lp{u}_0 \cdot \lp{u}_1} u'$.

        \item
        If $u \in S_{\lp{v}_0} = S_{\lp{u}_0}$ and $u' \in S_{\lp{v}_1} = S_{\lp{u}_1}$, then $u \leq_{\lp{u}_0 \cdot \lp{u}_1} u'$ immediately.
    \end{itemize}
    The case where $u \in S_{\lp{u}_1}$ and $u' \in S_{\lp{u}_0}$ can be disregarded, for there we find that $u' \leq_{\lp{v}_0 \cdot \lp{v}_1} u$, and thus $u' \leq_\lp{u} u$, meaning that $u = u'$ and contradicting disjointness of $\lp{u}_0$ and $\lp{u}_1$.

    \item
    For the labeling, let $u \in S_\lp{u}$.
    If $u \in S_{\lp{v}_i}$ for $i \in 2$, then $\lambda_\lp{u}(u) = \lambda_{\lp{u}_i}(u) = \lambda_{\lp{v}_i}(u) = \lambda_{\lp{v}_0 \cdot \lp{v}_1}(u)$.
\end{itemize}

We also claim that for $i \in 2$, it holds that $\lp{u}_i \sqsubseteq \lp{v}_i$, as witnessed by the identity function $S_{\lp{v}_i} \to S_{\lp{u}_i}$.
To see this, let $v, v' \in S_{\lp{v}_i}$ be such that $v \leq_{\lp{v}_i} v'$.
We then know that $v \leq_{\lp{v}_0 \cdot \lp{v}_1} v'$, and thus $v \leq_\lp{u} v'$ by the premise.
However, since $v, v' \in S_{\lp{v}_i} = S_{\lp{u}_i}$, it follows that $v \leq_{\lp{u}_i} v'$.

The first claim is now satisfied by choosing $V_0 = [\lp{v}_0]$ and $V_1 = [\lp{v}_1]$.
The second claim can be proved analogously; here, we split up $V = [\lp{v}]$ according to $U_0 = [\lp{u}_0]$ and $U_1 = [\lp{u}_1]$.
\end{proof}

\subsection{The generalized versions of Levi's lemma}

To prove \Cref{lemma:pomset-levi-generalised}, we first prove a simpler statement.

\begin{lemma}%
\label{lemma:pomset-levi}
Let $U, V, W, X$ be pomsets such that $U \cdot V \sqsubseteq W \cdot X$.
There exists a pomset $Y$ such that either $U \sqsubseteq W \cdot Y$ and $Y \cdot V \sqsubseteq X$, or $U \cdot Y \sqsubseteq W$ and $V \sqsubseteq Y \cdot X$.
Moreover, if $U$ and $V$ are series-parallel, then so is $Y$.
\end{lemma}
\begin{proof}
By \Cref{lemma:pomset-factorise-subsumption}, we find pomsets $W'$ and $X'$ with $W' \sqsubseteq W$ and $X' \sqsubseteq X$, such that $U \cdot V = W' \cdot X'$.
Let $\lp{u}, \lp{v}, \lp{w}', \lp{x}'$ be labelled posets such that $U = [\lp{u}]$, $V = [\lp{v}]$, $W' = [\lp{w'}]$ and $X = [\lp{x}']$.
Without loss of generality, we can assume that $\lp{u}$ is disjoint from $\lp{v}$, and $\lp{w}'$ from $\lp{x}'$, and that $\lp{u} \cdot \lp{v} = \lp{w}' \cdot \lp{x}'$.
Note that this means that $S_\lp{u} \cup S_\lp{v} = S_{\lp{w}'} \cup S_{\lp{x}'}$.

Suppose, towards a contradiction, that $S_\lp{u} \not\subseteq S_{\lp{w}'}$ and $S_{\lp{w}'} \not\subseteq S_\lp{u}$.
Then there exists a $u \in S_\lp{u} \setminus S_{\lp{w}'}$ and a $w \in S_{\lp{w}'} \setminus S_\lp{u}$.
Since $u \not\in S_{\lp{w}'}$, it follows that $u \in S_{\lp{x}'}$; by the same reasoning, we find that $w \in S_\lp{v}$.
But then $u \leq_{\lp{u} \cdot \lp{v}} w$, and $w \leq_{\lp{w}' \cdot \lp{x}'} u$, and since $\leq_{\lp{u} \cdot \lp{v}}$ and $\leq_{\lp{w}' \cdot \lp{x}'}$ coincide, we find that $u = w$ by antisymmetry; this is a contradiction, since $u \in S_\lp{u}$ and $w \not\in S_\lp{u}$.
Thus, either $S_\lp{u} \subseteq S_\lp{w}'$ or $S_\lp{w}' \subseteq S_\lp{u}$.

For the remainder of this proof, suppose that $S_\lp{u} \subseteq S_{\lp{w}'}$; we can prove the claim when $S_\lp{u} \supseteq S_{\lp{w}'}$ using similar arguments.
We choose $S = S_{\lp{w}'} \setminus S_\lp{u}$ and $\lp{y} = \lp{w'} \restr{S}$.
We now claim that $\lp{w}' = \lp{u} \cdot \lp{y}$.
To see this, we show that their carriers, orders and labellings coincide.
\begin{itemize}
    \item
    For the carrier, note that $\lp{u}$ and $\lp{y}$ are disjoint, and that $S_{\lp{w}'} = S_\lp{u} \cup (S_\lp{w'} \setminus S_\lp{u}) = S_\lp{u} \cup S_\lp{y}$.

    \item
    For the order, suppose first that $w_0, w_1 \in S_{\lp{w}'}$ with $w_0 \leq_{\lp{w'}} w_1$.
    There are two cases to consider.
    \begin{itemize}
        \item
        If $w_0, w_1 \in S_\lp{u}$ or $w_0, w_1 \in S_\lp{y}$, then $w_0 \leq_\lp{u} w_1$ or $w_0 \leq_\lp{y} w_1$, and thus $w_0 \leq_{\lp{u} \cdot \lp{y}} w_1$.
        \item
        If $w_0 \in S_\lp{u}$ and $w_1 \in S_\lp{y}$, then $w_0 \leq_{\lp{u} \cdot \lp{y}} w_1$ by definition.
    \end{itemize}
    The case where $w_1 \in S_\lp{u}$ and $w_0 \in S_\lp{y}$ can be discounted, for here we find that $w_0 \in S_\lp{y} \subseteq S_\lp{v}$, and thus $w_1 \leq_{\lp{u} \cdot \lp{v}} w_0$, meaning that $w_1 \leq_{\lp{w}' \cdot \lp{x}'} w_0$, which in turn implies that $w_0 = w_1$, contradicting that $S_\lp{u}$ and $S_\lp{y}$ are disjoint.

    Now suppose that $w_0, w_1 \in S_{\lp{w}'}$ with $w_0 \leq_{\lp{u} \cdot \lp{y}} w_1$.
    There are three cases to consider.
    \begin{itemize}
        \item
        If $w_0, w_1 \in \S_\lp{u}$, then $w_0 \leq_{\lp{u}} w_1$, and thus $w_0 \leq_{\lp{u} \cdot \lp{v}} w_1$.
        Since $\lp{u} \cdot \lp{v} = \lp{w}' \cdot \lp{x}'$, we have that $w_0 \leq_{\lp{w}' \cdot \lp{x}'} w_1$.
        Since $w_0, w_1 \in S_{\lp{w}'}$, we have $w_0 \leq_{\lp{w}'} w_1$.

        \item
        If $w_0, w_1 \in S_\lp{y}$, then $w_0 \leq_\lp{y} w_1$. Since $\leq_\lp{y}\ \subseteq\ \leq_{\lp{w}'}$, we find that $w_0 \leq_{\lp{w}'} w_1$.

        \item
        If $w_0 \in S_\lp{u}$ and $w_1 \in S_\lp{y}$, then $w_1 \in S_\lp{v}$ and therefore $w_0 \leq_{\lp{u} \cdot \lp{v}} w_1$.
        Since $\lp{u} \cdot \lp{v} = \lp{w}' \cdot \lp{x}'$, we have that $w_0 \leq_{\lp{w}' \cdot \lp{x}'} w_1$.
        Since $w_0, w_1 \in S_{\lp{w}'}$, we then know that $w_0 \leq_{\lp{w}'} w_1$.
    \end{itemize}

    \item
    For the labelling, let $w \in S_{\lp{w}'}$.
    If $w \in S_\lp{u}$, then $\lambda_{\lp{w}'}(w) = \lambda_{\lp{w'} \cdot \lp{x}'}(w) = \lambda_{\lp{u} \cdot \lp{v}}(w) = \lambda_{\lp{u}}(w) = \lambda_{\lp{u} \cdot \lp{y}}(w)$.
    Otherwise, if $w \not\in S_\lp{u}$, then $\lambda_{\lp{w}'}(w) = \lambda_\lp{y}(w)$ by definition of $\lp{y}$.
\end{itemize}

We now claim that $\lp{v} = \lp{y} \cdot \lp{x}'$.
To this end, we show that their carriers, orders and labellings coincide.
\begin{itemize}
    \item
    For the carrier, note that $S_\lp{y} \subseteq S_{\lp{w}'}$, and thus $S_\lp{y}$ is disjoint from $S_{\lp{x}'}$.
    Furthermore,
    \[
    S_{\lp{y} \cdot \lp{x}'} = S_\lp{y} \cup S_{\lp{x}'} = (S_{\lp{w}'} \setminus S_\lp{u}) \cup S_{\lp{x}'} = (S_{\lp{w}'} \cup S_{\lp{x}'}) \setminus S_\lp{u} = (S_\lp{u} \cup S_\lp{v}) \setminus S_\lp{u} = S_\lp{v}
    \]

    \item
    For the order, suppose first that $v_0, v_1 \in S_\lp{v}$ with $v_0 \leq_\lp{v} v_1$.
    Then $v_0 \leq_{\lp{u} \cdot \lp{v}} v_1$, and thus $v_0 \leq_{\lp{w}' \cdot \lp{x}'} v_1$.
    There are three cases to consider.
    \begin{itemize}
        \item
        If $v_0, v_1 \in S_\lp{y}$, then $v_0 \leq_{\lp{w}'} v_1$; since $\lp{w}' = \lp{u} \cdot \lp{y}$, we have that $v_0 \leq_\lp{y} v_1$, and thus $v_0 \leq_{\lp{y} \cdot \lp{x}'} v_1$.
        \item
        If $v_0, v_1 \in S_{\lp{x}'}$, then $v_0 \leq_{\lp{x}'} v_1$, and thus $v_0 \leq_{\lp{y} \cdot \lp{x}'} v_1$.
        \item
        If $v_0 \in S_\lp{y}$ and $v_1 \in S_{\lp{x}'}$, then $v_0 \leq_{\lp{y} \cdot \lp{x}'} v_1$ immediately.
    \end{itemize}
    The case where $v_1 \in S_\lp{y}$ and $v_0 \in S_{\lp{x}'}$ can be discounted, for here we find that $v_1 \in S_{\lp{w}'}$, and thus $v_1 \leq_{\lp{w}' \cdot \lp{x}'} v_0$, which would imply that $v_0 = v_1$, contradicting that $S_\lp{y}$ and $S_{\lp{x}'}$ are disjoint.

    Now suppose that $v_0, v_1 \in S_\lp{v}$ with $v_0 \leq_{\lp{y} \cdot \lp{x}'} v_1$.
    There are three cases to consider.
    \begin{itemize}
        \item
        If $v_0, v_1 \in S_\lp{y}$, then $v_0, v_1 \in S_{\lp{w}'}$.
        We then have that $v_0 \leq_{\lp{w}'} v_1$, and thus that $v_0 \leq_{\lp{w}' \cdot \lp{x}'} v_1$.
        Since $\lp{w}' \cdot \lp{x}' = \lp{u} \cdot \lp{v}$, we have that $v_0 \leq_{\lp{u} \cdot \lp{v}} v_1$, and since $v_0, v_1 \in S_\lp{v}$, it follows that $v_0 \leq_\lp{v} v_1$.

        \item
        If $v_0, v_1 \in S_{\lp{x}'}$, then $v_0, v_1 \not\in S_{\lp{w}'}$, and thus, since $S_\lp{u} \subseteq S_{\lp{w}'}$, it follows that $v_0, v_1 \not\in S_\lp{u}$.
        Since $v_0 \leq_{\lp{w}' \cdot \lp{x}'} v_1$ and thus $v_0 \leq_{\lp{u} \cdot \lp{v}} v_1$, we have $v_0 \leq_\lp{v} v_1$.

        \item
        If $v_0 \in S_\lp{y}$ and $v_1 \in S_{\lp{x}'}$, then $v_0 \in S_{\lp{w}'}$ and thus $v_0 \leq_{\lp{w}' \cdot \lp{x}'} v_1$, meaning that $v_0 \leq_{\lp{u} \cdot \lp{v}} v_1$.
        Since $v_0, v_1 \in S_\lp{v}$, this means that $v_0 \leq_\lp{v} v_1$.
    \end{itemize}

    \item
    For the labelling, let $v \in S_\lp{v}$.
    If $v \in S_\lp{y}$, then $\lambda_{\lp{v}}(v) = \lambda_{\lp{u} \cdot \lp{v}}(v) = \lambda_{\lp{w}' \cdot \lp{x}'}(v) = \lambda_{\lp{w}'}(v) = \lambda_\lp{y}(v) = \lambda_{\lp{y} \cdot \lp{x}'}(v)$.
    Otherwise, if $v \in S_{\lp{x}'}$, then $\lambda_{\lp{v}}(v) = \lambda_{\lp{u} \cdot \lp{v}}(v) = \lambda_{\lp{w}' \cdot \lp{x}'}(v) = \lambda_{\lp{x}'}(v) = \lambda_{\lp{y} \cdot \lp{x}'}(v)$.
\end{itemize}

We now choose $Y = [\lp{y}]$ to find that $W' = U \cdot Y$ and $V = Y \cdot X'$.
But then, since $W' \sqsubseteq W$ and $X' \sqsubseteq X$, we find that $U \cdot Y \sqsubseteq W$ and $V \sqsubseteq Y \cdot X$, fulfilling the first part of the claim.
Lastly, note that if $U \cdot V$ is series-parallel, it is \N-free.
This means that $W'$ must also be \N-free, since any \N that would occur in $W'$ would also occur in $U \cdot V$.
Because $Y$ is constructed as a sub-pomset of $W'$, it follows that $Y$ must also be \N-free, and thus by \Cref{lemma:pomset-sp-nfree} we find that $Y$ is series-parallel.
\end{proof}

\pomsetlevigeneralised*
\begin{proof}
The proof proceeds by induction on $n$.
In the base, where $n = 1$, we choose $m = 0$, $Y = U$ and $Z = V$ to satisfy the claim.

In the inductive step, assume the claim holds for $n-1$. We can write $U \cdot V = W_0 \cdot (W_1 \cdot W_2 \cdots W_{n-1})$.
By \Cref{lemma:pomset-levi}, there are two cases to consider.
\begin{itemize}
    \item
    Suppose that $X$ is a pomset such that $U \sqsubseteq W_0 \cdot X$ and $X \cdot V \sqsubseteq W_1 \cdot W_2 \cdots W_{n-1}$.
    By induction, we find $1 \leq m < n$ and pomsets $Y, Z$ such that $Y \cdot Z \sqsubseteq W_m$ and $X \sqsubseteq W_1 \cdot W_2 \cdots W_{m-1} \cdot Y$ and $V \sqsubseteq Z \cdot W_{m+1} \cdot W_{m+2} \cdots W_n$.
    Since in this case $U \sqsubseteq W_0 \cdot W_1 \cdots W_{m-1} \cdot X$, the claim follows.
    Moreover, if $U$ and $V$ are series-parallel, then so are $Y$ and $Z$, by induction.

    \item
    Suppose that $X$ is a pomset such that $U \cdot X \sqsubseteq W_0$ and $V \sqsubseteq X \cdot W_1 \cdot W_2 \cdots W_{n-1}$.
    We can then choose $m = 0$, $Y = U$ and $Z = X$ to satisfy the claim.
    Moreover, if $U$ and $V$ are series-parallel, then $X$ is series-parallel, meaning that $Y$ and $Z$ are also series-parallel.
    \qedhere
\end{itemize}
\end{proof}

\pomsetleviparallel*
\begin{proof}
Let $U = [\lp{u}]$, $V = [\lp{v}]$, $W = [\lp{w}]$, and $X = [\lp{x}]$, and assume without loss of generality that $\lp{u}$ and $\lp{v}$ as well as $\lp{w}$ and $\lp{x}$ are disjoint, and that $\lp{u} \parallel \lp{v} = \lp{w} \parallel \lp{x}$.
We can then choose $\lp{y}_0 = \lp{u} \restr{\lp{w}}$, $\lp{y}_1 = \lp{u} \restr{\lp{x}}$, $\lp{z}_0 = \lp{v} \restr{\lp{w}}$ and $\lp{z}_1 = \lp{v} \restr{\lp{x}}$.
We can then show that $\lp{u} = \lp{y}_0 \parallel \lp{y}_1$, $\lp{v} = \lp{z}_0 \parallel \lp{z}_1$, $\lp{w} = \lp{y}_0 \parallel \lp{z}_0$ and $\lp{x} = \lp{y}_1 \parallel \lp{z}_1$ by the usual technique, where for the last two equalities we use that $\lp{u} \restr{\lp{w}} = \lp{w} \restr{\lp{u}}$, $\lp{u} \restr{\lp{x}} = \lp{x} \restr{\lp{u}}$, $\lp{v} \restr{\lp{w}}$ and $\lp{v} \restr{\lp{x}} = \lp{x} \restr{\lp{v}}$.
The claim is then satisfied by choosing $Y_i = [\lp{y}_i]$ and $Z_i = [\lp{z}_i]$ for $i \in 2$.
\end{proof}

\subsection{The interpolation lemma}

\pomsetinterpolation*
\begin{proof}
Let $U = [\lp{u}]$, $V = [\lp{v}]$, $W = [\lp{w}]$ and $X = [\lp{x}]$, and assume without loss of generality that $\lp{u}$ and $\lp{v}$ are disjoint, as well as $\lp{w}$ and $\lp{x}$, and that $S_\lp{u} \cup S_\lp{v} = S_\lp{w} \cup S_\lp{x}$, such that the subsumption $\lp{u} \cdot \lp{v} \sqsubseteq \lp{w} \parallel \lp{x}$ is witnessed by the identity $i: S_\lp{w} \cup S_\lp{x} \to S_\lp{u} \cup S_\lp{v}$.

We choose labelled posets $\lp{w}_0$, $\lp{w}_1$, $\lp{x}_0$ and $\lp{x}_1$ as follows:
\begin{align*}
\lp{w}_0 &= \lp{w} \restr{S_\lp{u} \cap S_\lp{w}} &
\lp{w}_1 &= \lp{w} \restr{S_\lp{v} \cap S_\lp{w}} &
\lp{x}_0 &= \lp{x} \restr{S_\lp{u} \cap S_\lp{x}} &
\lp{x}_1 &= \lp{x} \restr{S_\lp{v} \cap S_\lp{x}}
\end{align*}
One easily verifies that these are pairwise disjoint.
To show that $\lp{u} \sqsubseteq \lp{w}_0 \parallel \lp{x}_0$, first note that
\[S_{\lp{w}_0 \parallel \lp{x}_0} = S_{\lp{w}_0} \cup S_{\lp{x}_0} = (S_\lp{u} \cap S_\lp{w}) \cup (S_\lp{u} \cap S_\lp{x}) = S_\lp{u} \cap (S_\lp{w} \cup S_\lp{x}) = S_\lp{u} \cap (S_\lp{u} \cup S_\lp{v}) = S_\lp{u}\]
We now claim that $i: S_{\lp{w}_0 \parallel \lp{x}_0} \to S_\lp{u}$, i.e., the identity on $S_\lp{u}$, is a subsumption witnessing that $\lp{u} \sqsubseteq \lp{p} \parallel \lp{q}$.
To see this, let $u_0, u_1 \in S_\lp{u}$ be such that $u_0 \leq_{\lp{w}_0 \parallel \lp{x}_0} u_1$.
If $u_0 \leq_{\lp{w}_0} z$, then $u_0 \leq_\lp{w} u_1$ by choice of $\lp{w}_0$.
But then $u_0 \leq_{\lp{w} \parallel \lp{x}} u_1$, and thus $u_0 \leq_{\lp{u} \cdot \lp{v}} u_1$ by the premise.
Since $u_0, u_1 \in S_\lp{u}$, we can conclude that $u_0 \leq_\lp{u} u_1$.
We can similarly show that $u_0 \leq_\lp{u} u_1$ when $u_0 \leq_{\lp{x}_0} z$ and thus conclude $\lp{u} \sqsubseteq \lp{w}_0 \parallel \lp{x}_0$.
The proof of $\lp{v} \sqsubseteq \lp{w}_1 \parallel \lp{x}_1$ is similar.

To see that $\lp{w}_0 \cdot \lp{w}_1 \sqsubseteq \lp{w}$, first note that $S_{\lp{w}_0 \cdot \lp{w}_1} = S_\lp{w}$ by reasoning similar to the above.
We claim that $i: S_\lp{w} \to S_{\lp{w}_0 \cdot \lp{w}_1}$, i.e., the identity on $S_\lp{w}$, is a subsumption witnessing that $\lp{w}_0 \cdot \lp{w}_1 \sqsubseteq \lp{w}$.
To see this, suppose that $w_0, w_1 \in S_\lp{w}$ such that $w_0 \leq_\lp{w} w_1$.
Then we know that $w_0 \leq_{\lp{w} \parallel \lp{x}} w_1$, and thus $w_0 \leq_{\lp{u} \cdot \lp{v}} w_1$ by the premise.
We can then exclude the case where $w_1 \in S_\lp{u}$ and $w_0 \in S_\lp{v}$, for then $w_1 \leq_{\lp{u} \cdot \lp{v}} w_0$ and thus $w_0 = w_1$ by antisymmetry, contradicting that $\lp{u}$ and $\lp{v}$ are disjoint.
Three cases remain to be considered.
\begin{itemize}
    \item If $w_0, w_1 \in S_\lp{u}$, then $w_0 \leq_{\lp{w}_0} w_1$, and thus $w_0 \leq_{\lp{w}_0 \cdot \lp{w}_1} w_1$.
    \item If $w_0, w_1 \in S_\lp{v}$, then $w_0 \leq_{\lp{w}_1} w_1$, and thus $w_0 \leq_{\lp{w}_0 \cdot \lp{w}_1} w_1$.
    \item If $w_0 \in S_\lp{u}$ and $w_1 \in S_\lp{v}$, then $w_0 \in S_{\lp{w}_0}$ and $w_1 \in S_{\lp{w}_1}$, thus $w_0 \leq_{\lp{w}_0 \cdot \lp{w}_1} w_1$ by definition.
\end{itemize}
Since $w_0 \leq_{\lp{w}_0 \cdot \lp{w}_1} w_1$ in all possible cases, we conclude that $i$ preserves ordering and is therefore a subsumption.
The proof that $\lp{x}_0 \cdot \lp{x}_1 \sqsubseteq \lp{x}$ is similar.

We can now choose $W_0 = [\lp{w}_0]$, $W_1 = [\lp{w}_1]$, $X_0 = [\lp{x}_0]$ and $X_1 = [\lp{x}_1]$ to satisfy the claim.
Moreover, we note that if $W$ and $X$ are series-parallel, then they are \N-free by \Cref{lemma:pomset-sp-nfree}.
The labelled posets $\lp{w}_0$, $\lp{w}_1$, $\lp{x}_0$ and $\lp{x}_1$ must then also be \N-free, and therefore $W_0$, $W_1$, $X_0$ and $X_1$ are series-parallel by \Cref{lemma:pomset-sp-nfree}.
\end{proof}

\subsection{The nullability function}

\nullable*
\begin{proof}
We start with the first claim.
This is shown by induction on $e$; we can disregard the cases where $\epsilon(e) = 0$, for then the claim holds trivially.
This leaves us with one case to consider in the base, namely $e = 1$; here we see that $\epsilon(e) = 1 \leqqaka 1 = e$.
For the inductive step, there are four cases to consider.
\begin{itemize}
    \item
    If $e = e_0 + e_1$ with $\epsilon(e) = 1$, then $\epsilon(e_i) = 1$ for some $i \in 2$.
    But then also $\epsilon(e) \leqqaka \epsilon(e_0) + \epsilon(e_1) \leqqaka e_0 + e_1 = e$.

    \item
    If $e = e_0 \cdot e_1$ with $\epsilon(e) = 1$, then $\epsilon(e_0) = \epsilon(e_1) = 1$.
    But then also $\epsilon(e) \leqqaka \epsilon(e_0) \cdot \epsilon(e_1) \leqqaka e_0 \cdot e_1 = e$.

    \item
    If $e = e_0 \parallel e_1$, then an argument similar to the above shows that $\epsilon(e) \leqqaka e$.

    \item
    If $e = e_0^\star$, then $\epsilon(e) = 1$.
    However, since $e = 1 + e_0 \cdot e$, we also have that $\epsilon(e) \leqqaka e$.
\end{itemize}

For the second claim, we observe that the direction from right to left follows from the first claim and \Cref{lemma:soundness}.
It remains to show the direction from left to right.
By \autoref{lemma:pomset-subsumption-base}, we know that if $1 \in \semaka{e}$, then $1 \in \sembka{e}$.
The proof proceeds by induction on $e$.
In the base, there is again only one case to consider, namely $e = 1$; the claim holds trivially here.
For the inductive step, there are four cases to consider.
\begin{itemize}
    \item
    If $e = e_0 + e_1$, then $1 \in \sembka{e_i}$ for some $i \in 2$.
    By induction, $\epsilon(e_i) = 1$, and thus $\epsilon(e) = 1$.

    \item
    If $e = e_0 \cdot e_1$, then there exist $U \in \sembka{e_0}$ and $V \in \sembka{e_1}$ such that $U \cdot V = 1$.
    By \autoref{lemma:pomset-unique-decomposition}, we have that $U = V = 1$, and thus by induction that $\epsilon(e_0) = \epsilon(e_1) = 1$.
    This implies that $\epsilon(e) = 1$.

    \item
    If $e = e_0 \parallel e_1$, then an argument similar to the above shows that $\epsilon(e) = 1$.

    \item
    If $e = e_0^\star$, then $\epsilon(e) = 1$ by definition.
    \qedhere
\end{itemize}
\end{proof}

\subsection{Observations about term width}

\widthvsnonempty*
\begin{proof}
The proof proceeds by induction on $e$.
In the base, we can disregard the cases where $e = 0$ or $e = 1$, where the claim holds vacuously.
This leaves us with the case where $e = a$ for some $a \in \alphabet$; here, the claim holds by definition of $\width{-}$.

In the inductive step, there are four cases to consider.
\begin{itemize}
    \item
    If $e = e_0 + e_1$, then either $U \in \sembka{e_0}$ or $U \in \sembka{e_1}$.
    In the former case, we find that $\width{e_0} > 0$ by induction, while in the latter case we find that $\width{e_1} > 0$ also by induction.
    This means that $\width{e} = \max(\width{e_0}, \width{e_1}) > 0$.

    \item
    If $e = e_0 \cdot e_1$, then there exist pomsets $U_0, U_1$ with $U = U_0 \cdot U_1$, such that $U_0 \in \sembka{e_0}$ and $U_1 \in \sembka{e_1}$.
    Since $U \neq 1$, we know that either $U_0 \neq 1$ or $U_1 \neq 1$.
    In the former case, we find that $\width{e_0} > 0$ by induction, while in the latter case we find that $\width{e_1} > 0$ also by induction.
    This means that $\width{e} = \max(\width{e_0}, \width{e_1}) > 0$.

    \item
    If $e = e_0 \parallel e_1$, then there exist pomsets $U_0, U_1$ with $U = U_0 \parallel U_1$, such that $U_0 \in \sembka{e_0}$ and $U_1 \in \sembka{e_1}$.
    Since $U \neq 1$, we know that either $U_0 \neq 1$ or $U_1 \neq 1$.
    In the former case, we find that $\width{e_0} > 0$ by induction, while in the latter case we find that $\width{e_1} > 0$ also by induction.
    This means that $\width{e} = \max(\width{e_0}, \width{e_1}) > 0$.

    \item
    If $e = e_0^\star$, then there exist pomsets $U_0, U_1, \dots, U_{n-1} \in \sembka{e_0}$ with $U = U_0 \cdot U_1 \cdots U_{n-1}$, such that for $0 \leq i < n$ we have that $U_i \in \sembka{e_0}$.
    Since $U \neq 1$, there exists an $i$ with $0 \leq i < n$ such that $U_i \neq 1$.
    By induction, we find that $\width{e_0} > 0$, which means that $\width{e} = \width{e_0} > 0$.
    \qedhere
\end{itemize}
\end{proof}

\widthvsequivalence*
\begin{proof}
If $e \equivbka 0 \equivbka f$, then $\width{e} = 0 = \width{f}$.
For the remaining cases, it suffices to verify the claim for all equivalences postulated for $\equivbka$ in \Cref{definition:equivalence}; that the claim is preserved by the congruence closure on these rules should be clear.

We first consider the base equivalences for $e \equivbka f$.
\begin{itemize}
    \item
    If $e = f + 0$, then $\width{e} = \max(\width{f}, 0) = \width{f}$.

    \item
    If $e = f + f$, then $\width{e} = \max(\width{f}, \width{f}) = \width{f}$.

    \item
    If $e = e_0 + e_1$ and $f = e_1 + e_0$, then
    \[\width{e} = \max(\width{e_0}, \width{e_1}) = \max(\width{e_1}, \width{e_0}) = \width{f}\]

    \item
    If $e = e_0 + (e_1 + e_2)$ and $f = (e_0 + e_1) + e_2$, then $\width{e} = \max(\width{e_0}, \width{e_1}, \width{e_2}) = \width{f}$.

    \item
    If $e = f \cdot 1$, then $\width{e} = \max(\width{f}, 0) = \width{f}$.
    The case where $f = e \cdot 1$ can be treated similarly.

    \item
    If $e = e' \cdot 0$ and $f = 0$, then $e \equivbka 0 \equivbka f$, and thus $\width{e} = 0 = \width{f}$.
    The case where $f = 0 \cdot f'$ and $e = 0$ can be treated similarly.

    \item
    If $e = e_0 \cdot (e_1 \cdot e_2)$ and $f = (e_0 \cdot e_1) \cdot e_2$, then $\width{e} = \max(\width{e_0}, \width{e_1}, \width{e_2}) = \width{f}$.

    \item
    If $e = e_0 \cdot (e_1 + e_2)$ and $f = e_0 \cdot e_1 + e_0 \cdot e_2$, then
    \[\width{e} = \max(e_0, \max(e_1, e_2)) = \max(\max(e_0, e_1), \max(e_0, e_2)) = \width{f}\]
    The case where $e = (e_0 + e_1) \cdot e_2$ and $f = e_0 \cdot e_2 + e_1 \cdot e_2$ can be treated similarly.

    \item
    If $e = e_0 \parallel e_1$ and $f = e_1 \parallel e_0$, then $\width{e} = \width{e_0} + \width{e_1} = \width{e_1} + \width{e_0} = \width{f}$.

    \item
    If $e = e' \parallel 0$ and $f = 0$, then $e \equivbka 0 \equivbka f$, and thus $\width{e} = 0 = \width{f}$.

    \item
    If $e = e_0 \parallel (e_1 \parallel e_2)$ and $f = (e_0 \parallel e_1) \parallel e_2$, then $\width{e} = \width{e_0} + \width{e_1} + \width{e_2} = \width{f}$.

    \item
    If $e = 1 + e_0 \cdot e_0^\star$ and $f = e_0^\star$, then $\width{e} = \max(0, \max(\width{e_0}, \width{e_0^\star})) = \width{e_0} = \width{f}$.
\end{itemize}

As for the inference rule, suppose that $e \leqqbka f$ with $e = e_0 + e_1 \cdot f$. (i.e., $e_0 + f_1 \cdot f + f \equivbka f$).
By induction $\max(\width{e_0}, \width{e_1}, \width{f}) = \width{f}$, and thus $\width{e_1^\star \cdot e_0} = \max(\width{e_1}, \width{e_0}) \leq \width{g}$.
From this, we can conclude that
\[\width{e_1^\star \cdot e_0 + f} = \max(\width{e_0}, \width{e_1}, \width{f}) = \width{f} \qedhere\]
\end{proof}

\subsection{Solutions to linear systems}

\linearsystemsolution*
\begin{proof}
Let $\AKA \in \{ \BKA, \CKA \}$.
We construct $x$ by induction on $|I|$.
In the base, $I = \emptyset$, meaning that the unique $I$-vector suffices as a least solution.

In the inductive step, let $k \in I$ and choose $I' = I - \{ k \}$.
We craft the $I'$-linear system $\ls{L}' = \angl{M', p'}$ as follows:
\begin{align*}
M'(i, j)
    &\triangleq M(i, k) \cdot {M(k, k)}^\star \cdot M(k, j) + M(i, j) \\
p'(i)
    &\triangleq p(i) + M(i, k) \cdot {M(k, k)}^\star \cdot p(k)
\intertext{%
Since $|I'| < |I|$, we know by induction that $\ls{L}'$ admits a least solution $x'$.
We construct the $I$-vector $x$ from $x'$ as follows:
}
x(i) &\triangleq
\begin{cases}
x'(i) & i \neq k \\
{M(k,k)}^\star \cdot \left(p(k) + \sum_{j \in I'} M(k, j) \cdot x'(j)\right) & i = k
\end{cases}
\end{align*}
We claim that $x$ is a solution of $\ls{L}$.
To see this, derive for $i \in I'$:
\begin{align*}
x(i)
    &\triangleq x'(i)
        \tag{Def. $x$} \\
    &\geqqaka p'(i) + \sum_{j \in I'} M'(i,j) \cdot x'(j)
        \tag{$x'$ solution of $\ls{L}'$} \\
    &\equivaka p(i) + M(i, k) \cdot {M(k, k)}^\star \cdot p(k) \\
    &\phantom{\equivaka} + \sum_{j \in I'} (M(i, k) \cdot {M(k, k)}^\star \cdot M(k, j) + M(i, j)) \cdot x'(j)
        \tag{Def. $\ls{L}'$} \\
    &\equivaka p(i) + \sum_{j \in I'} M(i, j) \cdot x'(j) \\
    &\phantom{\equivaka} + M(i, k) \cdot {M(k, k)}^\star \cdot \left( p(k) + \sum_{j \in I'} M(k, j) \cdot x'(j) \right)
        \tag{Distributivity} \\
    &\equivaka p(i) + \sum_{j \in I'} M(i, j) \cdot x(j) + M(i, k) \cdot x(k)
        \tag{Def. $x$} \\
    &\equivaka p(i) + \sum_{j \in I} M(i, j) \cdot x(j)
        \tag{Merge sum}
\end{align*}
Also, for $k$, we derive:
\begin{align*}
x(k)
    &\triangleq {M(k,k)}^\star \cdot \left(p(k) + \sum_{j \in I'} M(k, j) \cdot x'(j)\right)
        \tag{Def. $x$} \\
    &\equivaka (1 + M(k,k) \cdot {M(k,k)}^\star) \cdot \left(p(k) + \sum_{j \in I'} M(k, j) \cdot x'(j)\right)
        \tag{Unrolling} \\
    &\equivaka p(k) + \sum_{j \in I'} M(k, j) \cdot x'(j) \\
    &\phantom{\equivaka} + M(k,k) \cdot {M(k,k)}^\star \cdot \left(p(k) + \sum_{j \in I'} M(k, j) \cdot x'(j)\right)
        \tag{Distributivity} \\
    &\equivaka p(k) + \sum_{j \in I'} M(k, j) \cdot x(j) + M(k,k) \cdot x(k)
        \tag{Def. $x$} \\
    &\equivaka p(k) + \sum_{j \in I} M(k, j) \cdot x(j)
        \tag{Merge sum}
\end{align*}
We then know that $M \cdot x + b \leqqka x$, making $x$ a solution.

It remains to show that $x$ is the least solution.
To this end, let $y$ be any solution of $\ls{L}$.
We choose the $I'$-vector $y'$ by setting $y'(i) \triangleq y(i)$.
We claim that $y'$ is a solution of $\ls{L}'$.
To see this, we first note that
\begin{align*}
y(k)
    &\geqqaka p(k) + \sum_{j \in I} M(k, j) \cdot y(j)
        \tag{$y$ solution of $\ls{L}$} \\
    &\equivaka p(k) + M(k, k) \cdot y(k) + \sum_{j \in I'} M(k, j) \cdot y(j)
        \tag{Split sum} \\
    &\geqqaka {M(k, k)}^\star \cdot \left( p(k) + \sum_{j \in I'} M(k, j) \cdot y(j) \right)
        &\tag{Fixpoint axiom}
\end{align*}
With this in hand, we can derive
\begin{align*}
y'(i)
    &\geqqaka p(i) + \sum_{j \in I} M(i, j) \cdot y(j)
        \tag{$y$ solution of $\ls{L}$} \\
    &\equivaka p(i) + M(i, k) \cdot y(k) + \sum_{j \in I'} M(i, j) \cdot y(j)
        \tag{Split sum} \\
    &\geqqaka p(i) + M(i, k) \cdot {M(k,k)}^\star \cdot \left( p(k) + \sum_{j \in I'} M(k, j) \cdot y(j) \right) \\
    &\phantom{\geqqaka} + \sum_{j \in I'} M(i, j) \cdot y(j)
        \tag{observation above} \\
    &\equivaka p(i) + M(i, k) \cdot {M(k,k)}^\star \cdot p(k) \\
    &\phantom{\equivaka} + \sum_{j \in I'} (M(i, k) \cdot {M(k, k)}^\star \cdot M(k, j) + M(i, j)) \cdot y(j)
        \tag{Distributivity} \\
    &\equivaka p'(i) + \sum_{j \in I'} M'(i, j) \cdot y(j)
        \tag{Def. $\ls{L}'$}
\end{align*}
Thus $y'$ is a solution of $\ls{L}'$; since $x'$ is the least solution of $\ls{L}'$, we know that $x' \leqqaka y'$.
We furthermore derive
\begin{align*}
y(k)
    &\geqqaka {M(k,k)}^\star \cdot \left( p(k) + \sum_{j \in I'} M(k,j) \cdot y(j) \right)
        \tag{observation above} \\
    &\equivaka {M(k,k)}^\star \cdot \left( p(k) + \sum_{j \in I'} M(k,j) \cdot y'(j) \right)
        \tag{Def. $y'$} \\
    &\geqqaka {M(k,k)}^\star \cdot \left( p(k) + \sum_{j \in I'} M(k,j) \cdot x'(j) \right)
        \tag{$y'$ solution of $\ls{L}'$} \\
    &\geqqaka {M(k,k)}^\star \cdot \left( p(k) + \sum_{j \in I'} M(k,j) \cdot x(j) \right)
        \tag{Def. $x'$} \\
    &\triangleq x(k)
        \tag{Def. $x$}
\end{align*}
In total, we find that $x \leqqaka y$, making $x$ the least solution of $\ls{L}$.

Finally, note that in all derivation steps, $\AKA$ could have been either $\BKA$ or $\CKA$; since the constructed least solution is the same regardless of the choice of $\AKA$, the final claim is also satisfied.
\end{proof}

\section{Proofs for \Cref{section:axiomatising}}%
\label{appendix:axiomatising-proofs}

\closurebase*
\begin{proof}
That $e \equivcka e$ is immediate from the fact that $\equivcka$ is a congruence.
It remains to show $\sembka{e} = \down{\sembka{e}}$.
For $e = 0$, this holds immediately, since $\sembka{e} = \emptyset$.
For $e = 1$ or $e = a$ for some $a \in \alphabet$, the claim follows from \Cref{lemma:pomset-subsumption-base}.
\end{proof}

\subsection{Parallel splitting}

\psplitfinite*
\begin{proof}
The proof proceeds by induction on $e$.
In the base, where $e = 0$, $e = 1$ or $e = a$ for some $a \in \alphabet$, the claim holds immediately: since only the first rule applies, $\psplit{e}$ only contains $\angl{e, 1}$ and $\angl{1, e}$.

For the inductive step, suppose that $\lefts \psplit{e} \rights$; one of five cases must hold.
\begin{itemize}
    \item
    $\lefts = e$ and $\rights = 1$, or $\lefts = 1$ and $\rights = e$.

    \item
    $e = e_0 + e_1$, with either $\lefts \psplit{e_0} \rights$, or $\lefts \psplit{e_1} \rights$.

    \item
    $e = e_0 \cdot e_1$, with an $i \in 2$ such that $\lefts \psplit{e_i} \rights$ and $\epsilon(e_{1-i}) = 1$.

    \item
    $e = e_0 \parallel e_1$, with $\lefts = \lefts_0 \parallel \lefts_1$ and $\rights = \rights_0 \parallel \rights_1$, such that $\lefts_i \psplit{e_i} \rights_i$ for all $i \in 2$.

    \item
    $e = e_0^\star$, with $\lefts \psplit{e_0} \rights$.
\end{itemize}

In all of these, there are only finitely many $\lefts, \rights \in \terms$ that satisfy the derived restrictions --- in the first, this is immediate, in the others it follows by induction.
We conclude that $\psplit{e}$ is finite.
\end{proof}

\psplitdomination*
\begin{proof}
The proof proceeds by induction on the construction of $\psplit{e}$.
In the base, either $\lefts = e$ and $\rights = 1$, or $\lefts = 1$ and $\rights = e$; in both cases, $\lefts \parallel \rights \equivbka e$, and so the claim follows.

For the inductive step, there are five cases to consider.
\begin{itemize}
    \item
    If $e = e_0 + e_1$ while $\lefts \psplit{e_i} \rights$ for some $i \in 2$, then by induction we know that $\lefts \parallel \rights \leqqbka e_i$.
    But since $e_i \leqqbka e$, it follows that $\lefts \parallel \rights \leqqbka e$.

    \item
    If $e = e_0 \cdot e_1$ while $\lefts \psplit{e_i} \rights$ and $\epsilon(e_{1-i}) = 1$ for some $i \in 2$, then by induction we know that $\lefts \parallel \rights \leqqbka e_i$.
    If $i = 0$, then $e_i \equivbka e_0 \cdot 1 \leqqbka e_0 \cdot e_1 = e$ (by \Cref{lemma:nullable}); if $e = 1$, we find $e_i \leqqbka e$ analogously.
    This allows us to conclude that $\lefts \parallel \rights \leqqbka e$.

    \item
    If $e = e_0 \parallel e_1$ and $\lefts = \lefts_0 \parallel \lefts_1$ and $\rights = \rights_0 \parallel \rights_1$ while $\lefts_i \psplit{e_i} \rights_i$ for all $i \in 2$, then by induction we know that $\lefts_i \parallel \rights_i \leqqbka e_i$ for all $i \in 2$.
    We can then derive that
    \[ \lefts \parallel \rights = (\lefts_0 \parallel \lefts_1) \parallel (\rights_0 \parallel \rights_1) \equivbka (\lefts_0 \parallel \rights_0) \parallel (\lefts_1 \parallel \rights_1) \leqqbka e_0 \parallel e_1 = e \]

    \item
    If $e = e_0^\star$ while $\lefts \psplit{e} \rights$, then $\lefts \parallel \rights \leqqbka e_0$ by induction.
    Since $e_0 \leqqbka e$, the claim follows. \qedhere
\end{itemize}
\end{proof}

\subsection{Sequential splitting}

\ssplitfinite*
\begin{proof}
The proof proceeds by induction on $e$.
In the base, we can disregard the case where $e = 0$, for no rule applies here.
This leaves us two cases to consider.
\begin{itemize}
    \item
    If $e = 1$, then $\ssplit{e} = \{ \angl{1, 1} \}$, which makes $\ssplit{e}$ finite.

    \item
    If $e = a$ for some $a \in \alphabet$, then $\ssplit{e} = \{ \angl{a, 1}, \angl{1, a} \}$, which makes $\ssplit{e}$ finite again.
\end{itemize}

In the inductive step, suppose that $\lefts, \rights \in \terms$ are such that $\lefts \ssplit{e} \rights$.
There are four cases to consider.
\begin{itemize}
    \item
    If $e = e_0 + e_1$, then $\lefts \ssplit{e_i} \rights$ for some $i \in 2$.

    \item
    If $e = e_0 \cdot e_1$, then either $\lefts = e_0 \cdot \lefts'$ and $\lefts' \ssplit{e_1} \rights$, or $\rights = \rights' \cdot e_1$ and $\lefts \ssplit{e_0} \rights'$.

    \item
    If $e = e_0 \parallel e_1$, then $\lefts = \lefts_0 \parallel \lefts_1$ and $\rights = \rights_0 \parallel \rights_1$, such that for $i \in 2$ it holds that $\lefts_i \ssplit{e_i} \rights_i$.

    \item
    If $e = e_0^\star$, then either $\lefts = \rights = 1$, or $\lefts = e \cdot \lefts'$ and $\rights = \rights' \cdot e$ such that $\lefts \ssplit{e_0} \rights$
\end{itemize}
In all cases, there are finitely many $\lefts, \rights \in \terms$ that satisfy the restrictions put on them, by induction.
\end{proof}

\ssplitdomination*
\begin{proof}
The proof proceeds by induction on the construction of $\ssplit{e}$. In the base, there are three cases to consider.
\begin{itemize}
    \item
    If $e = \lefts = \rights = 1$, then $\lefts \cdot \rights \equivcka e$, and so the claim holds immediately.
    \item
    If $e = a$, and either $\lefts = 1$ and $\rights = a$, or $\lefts = a$ and $\rights = 1$, then $\lefts \cdot \rights \equivcka e$.
    \item
    If $e = e_0^\star$ and $\lefts = \rights = 1$, then $\lefts \cdot \rights \equivcka 1 \leqqcka e$, and so the claim holds.
\end{itemize}

For the inductive step, there are four cases to consider.
\begin{itemize}
    \item
    If $e = e_0 + e_1$ and $\lefts \ssplit{e_i} \rights$ for some $i \in 2$, then $\lefts \cdot \rights \leqqcka e_i$ by induction.
    Since $e_i \leqqcka e$, the claim then follows.

    \item
    If $e = e_0 \cdot e_1$ and $\rights = \rights' \cdot e_1$ with $\lefts \ssplit{e_0} \rights'$, then by induction we find that $\lefts \cdot \rights' \leqqcka e_0$.
    It then follows that $\lefts \cdot \rights = \lefts \cdot \rights' \cdot e_1 \leqqcka e_0 \cdot e_1 = e$.
    The case where $e = e_0 \cdot e_1$ and $\lefts = e_0 \cdot \lefts'$ with $\lefts' \ssplit{e_1} \rights$ can be treated similarly.

    \item
    If $e = e_0 \parallel e_1$ and $\lefts = \lefts_0 \parallel \lefts_1$ and $\rights = \rights_0 \parallel \rights_1$ such that $\lefts_i \ssplit{e_i} \rights_i$ for all $i \in 2$, then by induction we have that $\lefts_i \cdot \rights_i \leqqcka e_i$.
    We then find that
    \[
    \lefts \cdot \rights = (\lefts_0 \parallel \lefts_1) \cdot (\rights_0 \cdot \rights_1) \leqqcka (\lefts_0 \cdot \rights_0) \parallel (\lefts_1 \cdot \rights_1) \leqqcka e_0 \parallel e_1 = e
    \]

    \item
    If $e = e_0^\star$ and $\lefts = e_0^\star \cdot \lefts'$ and $\rights = \rights' \cdot e_0^\star$ such that $\lefts' \ssplit{e_0} \rights'$, then by induction we have that $\lefts' \cdot \rights' \leqqcka e_0$.
    This allows us to derive that $\lefts \cdot \rights = e_0^\star \cdot \lefts' \cdot \rights' \cdot e_0^\star \leqqcka e_0^\star \cdot e_0 \cdot e_0^\star \leqqcka e_0^\star = e$.
    \qedhere
\end{itemize}
\end{proof}

\subsection{Right-hand remainders}

\remaindersfinite*
\begin{proof}
Let $R^+(e)$ denote $R(e) \setminus \{ e \}$.
We first prove a number of auxiliary claims, to wit:
\begin{enumerate}[(i)]
    \item $R^+(0) = \emptyset$
    \item $R^+(1) = \{ 1 \}$
    \item for $a \in \alphabet$, it holds that $R^+(a) = \{ a, 1 \}$
    \item for $e, f \in \terms$, it holds that $R^+(e + f) \subseteq R(e) \cup R(f)$.
    \item for $e, f \in \terms$, it holds that $R^+(e \cdot f) \subseteq \{ e' \cdot f : e' \in R(e) \} \cup R(f)$
    \item for $e, f \in \terms$, it holds that $R^+(e \parallel f) \subseteq \{ e' \parallel f' : e' \in R(e), f' \in R(f) \}$
    \item for $e \in \terms$, it holds that $R^+(e^\star) \subseteq \{ 1, e^\star \} \cup \{ e' \cdot e^\star : e' \in R(e) \} $
\end{enumerate}
To prove a claim of the form $R^+(g) \subseteq T$ for some $g \in \terms$ and $T \subseteq \terms$, it suffices to show that if $\lefts, \rights \in \terms$ such that $\lefts \ssplit{g} \rights$, then $\rights \in T$, and moreover that $T$ is closed under taking right-remainders, i.e., if $h \in T$ and $\lefts, \rights \in \terms$ such that $\lefts \ssplit{h} \rights$, then $\rights \in T$.
We treat the claims one-by-one.
\begin{enumerate}[(i)]
    \item\label{subclaim-remainder:case-zero}
    If $g = 0$ and $T = \emptyset$, then the claim holds vacuously --- there are no $\lefts, \rights \in \terms$ such that $\lefts \ssplit{0} \rights$, and $\emptyset$ is immediately closed under taking right-remainders.

    \item\label{subclaim-remainder:case-one}
    If $g = 1$ and $T = \{ 1 \}$, suppose that $\lefts, \rights \in \terms$ such that $\lefts \ssplit{1} \rights$.
    By definition of $\ssplit{1}$, we then find that $\lefts = \rights = 1$; it then follows that $\rights \in T$.
    By the same argument, $T$ is closed under taking right-remainders.

    \item\label{subclaim-remainder:case-primitive}
    If $g = a$ and $T = \{ a, 1 \}$, suppose that $\lefts, \rights \in \terms$ such that $\lefts \ssplit{a} \rights$.
    By definition of $\ssplit{a}$, we then find that either $\lefts = 1$ and $\rights = a$, or $\lefts = a$ and $\rights = 1$; in both cases, $\rights \in T$.
    By an argument similar to the above, as well as the reasoning for the previous case, $T$ is closed under taking right-remainders.

    \item\label{subclaim-remainder:case-choice}
    If $g = e + f$ and $T = R(e) \cup R(f)$, suppose that $\lefts, \rights \in \terms$ such that $\lefts \ssplit{e + f} \rights$.
    By definition of $\ssplit{e + f}$, we then find that either $\lefts \ssplit{e} \rights$ or $\lefts \ssplit{f} \rights$.
    In the former case, $\rights \in R(e)$, while in the latter case $\rights \in R(f)$; in either case, $\rights \in T$.
    Lastly, $T$ is closed under taking right-remainders because both $R(e)$ and $R(f)$ are, individually.

    \item\label{subclaim-remainder:case-sequential}
    If $g = e \cdot f$ and $T = \{ e' \cdot f : e' \in R(e) \} \cup R(f)$, suppose that $\lefts, \rights \in \terms$ such that $\lefts \ssplit{e \cdot f} \rights$.
    By definition of $\ssplit{e \cdot f}$, we then find that either $\lefts = e \cdot \lefts'$ and $\lefts' \ssplit{f} \rights$, or that $\rights = \rights' \cdot f$ and $\lefts \ssplit{e} \rights'$.
    In the former case, $\rights \in R(f)$; in the latter case, $\rights' \in R(e)$, and thus $\rights \in \{ e' \cdot f : e' \in R(e) \}$; in either case, $\rights \in T$.

    To see that $T$ is closed under taking right-remainders, it suffices to consider the case where $h = e' \cdot f$ for some $e' \in R(e)$.
    If $\lefts, \rights \in \terms$ are such that $\lefts \ssplit{h} \rights$, then either $\lefts = e' \cdot \lefts'$ and $\lefts' \ssplit{f} \rights$, or $\rights = \rights' \cdot f$ and $\lefts \ssplit{e'} \rights'$.
    In the former case, $\rights \in R(f)$, while in the latter case $\rights' \in R(e') \subseteq R(e)$, and thus $\rights \in \{ e' \cdot f : e' \in R(e) \}$; in either case, $\rights \in T$.

    \item\label{subclaim-remainder:case-parallel}
    If $g = e \parallel f$ and $T = \{ e' \parallel f' : e' \in R(e), f' \in R(f) \}$, suppose that $\lefts, \rights \in \terms$ such that $\lefts \ssplit{e \parallel f} \rights$.
    By definition of $\ssplit{e \parallel f}$, we find that $\lefts = \lefts_e \parallel \lefts_f$ and $\rights = \rights_e \parallel \rights_f$ such that $\lefts_e \ssplit{e} \rights_e$ and $\lefts_f \ssplit{f} \rights_f$.
    In that case, $\rights_e \in R(e)$ and $\rights_f \in R(f)$, and thus $\rights \in T$.

    To see that $T$ is closed under taking right-remainders, an argument similar to the above applies.

    \item\label{subclaim-remainder:case-star}
    If $g = e^\star$ and $T = \{ 1, e^\star \} \cup \bigcup_{e' \in R(e)} R(e' \cdot e^\star)$, suppose that $\lefts, \rights \in \terms$ such that $\lefts \ssplit{e^\star} \rights$.
    By definition of $\ssplit{e^\star}$, we find that either $\lefts = \rights = 1$, or $\lefts = e^\star \cdot \lefts'$ and $\rights = \rights' \cdot e^\star$ with $\lefts' \ssplit{e} \rights'$.
    In the former case, $\rights \in T$ immediately; in the latter case, we find that $\rights' \in R(e)$, and thus $\rights \in \{ e' \cdot e^\star : e' \in R(e) \} \subseteq T$.

    To see that $T$ is closed under taking right-remainders, note that the case for $h = 1$ is covered by~\eqref{subclaim-remainder:case-one}, and the case where $h = e^\star$ is discussed above.
    It therefore suffices to consider the case where $h = e' \cdot e^\star$ for some $e' \in R(e)$.
    Suppose that $\lefts, \rights \in \terms$ such that $\lefts \ssplit{e' \cdot e^\star} \rights$; by definition of $\ssplit{e' \cdot e^\star}$, we know that either $\lefts = e' \cdot \lefts'$ and $\lefts' \ssplit{e^\star} \rights$, or $\rights = \rights' \cdot e^\star$ and $\lefts \ssplit{e'} \rights'$.
    In the former case, $\rights \in T$ by the argument for $g = e^\star$ above.
    In the latter case, $\rights = \rights' \cdot e^\star \in \{ e'' \cdot e^\star : e'' \in R(e') \} \subseteq \{ e'' \cdot e^\star : e'' \in R(e) \} \subseteq T$.
\end{enumerate}

We can use these observations to show that $R(e) = R^+(e) \cup \{ e \}$ is finite, by induction on $e$.
In the base, where $e = 0$, $e = 1$ or $e = a$, we have that $R(e)$ is finite by~\eqref{subclaim-remainder:case-one}--\eqref{subclaim-remainder:case-primitive}.
In the inductive step, assume that the claim holds for all proper subterms of $e$.
We now have that $e = e_0 + e_1$, $e = e_0 \cdot e_1$, $e = e_0 \parallel e_1$ or $e = e_0^\star$ for some $e_0, e_1 \in \terms$.
It then follows that $R(e)$ is finite by~\eqref{subclaim-remainder:case-choice}--\eqref{subclaim-remainder:case-star} and the induction hypothesis.
\end{proof}

\fi%

\section{Worked example: a non-trivial closure}%
\label{appendix:example-solve}

In this appendix, we solve an instance of a linear system as defined in \Cref{definition:closure-system} for a given parallel composition.
For the sake of brevity, the steps are somewhat coarse-grained; the reader is encouraged to reproduce the steps by hand.

Consider the expression $e \parallel f = a^* \parallel b$.
The linear system $\ls{L}_{e, f}$ that we obtain from this expression consists of six inequations; in matrix form (with zeroes omitted), this system is summarised as follows:%
\footnote{%
    Actually, the system obtained from $a^\star \parallel b$ as a result of \Cref{definition:closure-system} is slightly larger; it also contains rows and columns labelled by $1 \cdot a^\star \parallel 1$ and $1 \cdot a^\star \parallel b$; these turn out to be redundant.
    We omit these rows from the example for simplicity.
}
\[
\begin{array}{*6{r}}
\color{gray} 1 \parallel 1 \\
\color{gray} 1 \parallel b\hspace{0.26mm} \\
\color{gray} a \cdot a^\star \parallel 1 \\
\color{gray} a^\star \parallel 1 \\
\color{gray} a \cdot a^\star \parallel b\hspace{0.26mm} \\
\color{gray} a^\star \parallel b\hspace{0.26mm}
\end{array}
\left(%
\setlength{\arraycolsep}{0.7em}
\begin{array}{*6{c};{2pt/2pt}c}
1 &   &   &   &   &   & 1 \\
b & 1 &   &   &   &   & b \\
a &   & a^\star & a \cdot a^\star &   &   & a \cdot a^\star \\
1 &   & a^\star & a^\star \cdot a &   &   & a^\star \\
a \parallel b & a & a^\star \parallel b & a \cdot a^\star \parallel b & a^\star & a \cdot a^\star & a \cdot a^\star \parallel b \\
b & 1 & a^\star \parallel b & a \cdot a^\star \parallel b & a^\star & a \cdot a^\star & a^\star \parallel b \\
\end{array}
\right)
\]

Let us proceed under the assumption that $x$ is a solution to the system; the constraint imposed on $x$ by the first two rows is given by the inequations
\begin{align}
x(1 \parallel 1) + 1 &\leqqcka x(1 \parallel 1) \\
b \cdot x(1 \parallel 1) + x(1 \parallel b) + b &\leqqcka x(1 \parallel b)
\end{align}
Because these inequations do not involve the other positions of the system, we can solve them in isolation, and use their solutions to find solutions for the remaining positions; it turns out that choosing $x(1 \parallel 1) = 1$ and $x(1 \parallel b) = b$ suffices here.

We carry on to fill these values into the inequations given by the third and fourth row of the linear system.
After some simplification, these work out to be
\begin{align}
a \cdot a^{\star} + a \cdot a^{\star} \cdot x(a^{\star} \parallel 1) + a^{\star} \cdot x(a \cdot a^{\star} \parallel 1) &\leqqcka x(a \cdot a^{\star} \parallel 1) \label{ineq:third-row}\\
a^{\star} + a^{\star} \cdot a \cdot x(a^{\star} \parallel 1) + a^{\star} \cdot x(a \cdot a^{\star} \parallel 1) &\leqqcka x(a^{\star} \parallel 1) \label{ineq:fourth-row}
\end{align}
Applying the least fixpoint axiom to~\eqref{ineq:third-row} and simplifying, we obtain
\begin{equation}
a \cdot a^\star + a \cdot a^\star \cdot x(a^\star \parallel 1) \leqqcka x(a \cdot a^\star \parallel 1)
\end{equation}
Substituting this into~\eqref{ineq:fourth-row} and simplifying, we find that
\begin{equation}
a^{\star} + a \cdot a^\star \cdot x(a^{\star} \parallel 1) \leqqcka x(a^{\star} \parallel 1)
\end{equation}
This inequation, in turn, gives us that $a^\star \leqqcka x(a^\star \parallel 1)$ by the least fixpoint axiom.
Plugging this back into~\eqref{ineq:third-row} and simplifying, we find that
\begin{equation}
a \cdot a^{\star} + a^{\star} \cdot x(a \cdot a^{\star} \parallel 1) \leqqcka x(a \cdot a^{\star} \parallel 1)
\end{equation}
Again by the least fixpoint axiom, this tells us that $a \cdot a^\star \leqqcka x(a \cdot a^\star \parallel 1)$.
One easily checks that $x(a \cdot a^\star \parallel 1) = a \cdot a^\star$ and $x(a^\star \parallel 1) = a^\star$ are solutions to~\eqref{ineq:third-row} and~\eqref{ineq:fourth-row}; by the observations above, they are also the least solutions.

It remains to find the least solutions for the final two positions.
Filling in the values that we already have, we find the following for the fifth row:
\begin{align}
a \parallel b + a \cdot b + (a^\star \parallel b) \cdot a \cdot a^\star + (a \cdot a^\star \parallel b) \cdot a^\star \hspace{1cm}\nonumber \\
\phantom{0} + a^\star \cdot x(a \cdot a^\star \parallel b) + a \cdot a^\star \cdot x(a^\star \parallel b) + a \cdot a^\star \parallel b &\leqqcka x(a \cdot a^\star \parallel b) \label{ineq:fifth-row}
\end{align}
Applying the exchange law%
\footnote{%
    A caveat here is that applying the exchange law indiscriminately may lead to a term that is not a closure (specifically, it may violate the semantic requirement in \Cref{definition:closure}).
    The algorithm used to solve arbitrary linear systems in \Cref{lemma:linear-system-solution} does not make use of the exchange law to simplify terms, and thus avoids this pitfall.
}
to the first three terms, we find that they are contained in $(a \cdot a^\star \parallel b) \cdot a^\star$, as is the last term;~\eqref{ineq:fifth-row} thus simplifies to
\begin{align}
(a \cdot a^\star \parallel b) \cdot a^\star + a^\star \cdot x(a \cdot a^\star \parallel b) + a \cdot a^\star \cdot x(a^\star \parallel b) &\leqqcka x(a \cdot a^\star \parallel b) \label{ineq:fifth-row-simplified}
\end{align}
By the least fixpoint axiom, we find that
\begin{align}
a^\star \cdot (a \cdot a^\star \parallel b) \cdot a^\star + a \cdot a^\star \cdot x(a^\star \parallel b) &\leqqcka x(a \cdot a^\star \parallel b) \label{ineq:fifth-row-fixpoint}
\end{align}

For the sixth row, we find that after filling in the solved positions, we have
\begin{align}
b + b + (a^\star \parallel b) \cdot a \cdot a^\star + (a \cdot a^\star \parallel b) \cdot a^\star \hspace{14mm}\nonumber\\
\phantom{0} + a^\star \cdot x(a \cdot a^\star \parallel b) + a \cdot a^\star \cdot x(a^\star \parallel b) + a^\star \parallel b &\leqqcka x(a^\star \parallel b)
\end{align}
Simplifying and applying the exchange law as before, it follows that
\begin{align}
(a^\star \parallel b) \cdot a^\star + a^\star \cdot x(a \cdot a^\star \parallel b) + a \cdot a^\star \cdot x(a^\star \parallel b) &\leqqcka x(a^\star \parallel b) \label{ineq:sixth-row-simplified}
\end{align}
We then subsitute~\eqref{ineq:fifth-row-fixpoint} into~\eqref{ineq:sixth-row-simplified} to find that
\begin{align}
(a^\star \parallel b) \cdot a^\star + a \cdot a^\star \cdot x(a^\star \parallel b) &\leqqcka x(a^\star \parallel b)
\end{align}
which, by the least fixpoint axiom, tells us that $a^\star \cdot (a^\star \parallel b) \cdot a^\star \leqqcka x(a^\star \parallel b)$.
Plugging the latter back into~\eqref{ineq:fifth-row-simplified}, we find that
\begin{align}
a^\star \cdot (a \cdot a^\star \parallel b) \cdot a^\star + a \cdot a^\star \cdot a^\star \cdot (a^\star \parallel b) \cdot a^\star &\leqqcka x(a \cdot a^\star \parallel b)
\end{align}
which can, using the exchange law, be reworked into
\begin{align}
a^\star \cdot (a \cdot a^\star \parallel b) \cdot a^\star &\leqqcka x(a \cdot a^\star \parallel b)
\end{align}

Now, if we choose $x(a \cdot a^\star \parallel b) = a^\star \cdot (a \cdot a^\star \parallel b) \cdot a^\star$ and $x(a^\star \parallel b) = a^\star \cdot (a^\star \parallel b) \cdot a^\star$, we find that these choices satisfy~\eqref{ineq:fifth-row-simplified} and~\eqref{ineq:sixth-row-simplified} --- making them part of a solution; by construction, they are also the least solutions.

In summary, $x$ is a solution to the linear system, and by construction it is also the least solution.
The reader is encouraged to verify that our choice of $x(a^\star \parallel b)$ is indeed a closure of $a^\star \parallel b$.

\bibliography{bibliography}

\end{document}